\documentclass[conference]{IEEEtran}
\usepackage[utf8]{inputenc}
\newcommand*{\xor}{\oplus}

\newcommand*{\zo}[1]{$z_{n+1}$}
\newcommand{\bfs}{\ensuremath{\mathsf{BFnS}}}
\newcommand{\ROBDD}{\ensuremath{\mathsf{ROBDD}}}
\newcommand{\FDD}{\ensuremath{\mathsf{FDD}}}
\newcommand{\snf}{\ensuremath{\mathsf{SynNNF}}}

\newcommand{\eqnf}{\ensuremath{\mathsf{SAUNF}}}
\newcommand{\eqnfs}{\ensuremath{\mathsf{SAUNF}}}
\newcommand{\wDNNF}{\ensuremath{\mathsf{wDNNF}}}
\newcommand{\DNNF}{\ensuremath{\mathsf{DNNF}}}
\newcommand{\dDNNF}{\ensuremath{\mathsf{dDNNF}}}
\newcommand{\CNF}{\ensuremath{\mathsf{CNF}}}
\newcommand{\DNF}{\ensuremath{\mathsf{DNF}}}

\newcommand{\bigO}[1]{\ensuremath{\mathcal{O}\big({#1}\big)}}
\newcommand{\Fc}{\ensuremath{G^{-}}}
\newcommand{\Gc}{\ensuremath{G^{+}}}

\usepackage{amsthm}
\usepackage{amsmath}
\usepackage{amssymb}
\usepackage{thm-restate}
\usepackage{tikz}
\usepackage{graphicx,algorithmic}
\usepackage{caption, subcaption}
\usepackage{wrapfig}
\usepackage{multirow}
\usepackage{multicol}
\usepackage[ruled,vlined,linesnumbered,lined]{algorithm2e}
\newtheorem{theorem}{Theorem}
\newtheorem{lemma}{Lemma}
\newtheorem{claim}{Claim}

\newtheorem{corollary}{Corollary}
\newtheorem{definition}{Definition}

\newtheorem{example}{Example}

\newcommand{\mset}[1]{\ensuremath{\mathsf{set}({#1})}}
\newcommand{\mseq}[1]{\ensuremath{\mathbf{{#1}}}}
\newcommand{\F}{\ensuremath{\bot}}
\newcommand{\T}{\ensuremath{\top}}
\newcommand{\ttrue}{\ensuremath{\mathsf{true}}}
\newcommand{\ffalse}{\ensuremath{\mathsf{false}}}

\newcommand{\Eup}[2]{\ensuremath{#1 \leftrightsquigarrow #2}}

\newcommand{\myP}{\ensuremath{\mathsf{P}}}
\newcommand{\myNP}{\ensuremath{\mathsf{NP}}}
\newcommand{\mycoNP}{\ensuremath{\mathsf{co}\mathsf{NP}}}
\newcommand{\myPSPoly}{\ensuremath{\mathsf{P}/\mathsf{Poly}}}

\newcommand{\R}[2]{R[#1,#2]}

\newcommand{\remove}[1]{}
\newcommand{\bpsi}{\ensuremath{\mathbf{\Psi}}}
\newcommand{\Genskolem}{SkGen}
\newcommand{\colora}{black}
\newcommand{\colorb}{black}
\newcommand{\mycolora}[1]{{\color{\colora}{#1}}}

\newcommand{\mycolorb}[1]{{\color{\colorb}{#1}}}
\newcommand{\colornega}{black}
\newcommand{\colornegb}{black}
\newcommand{\mycolornega}[1]{{\color{\colornega}{#1}}}

\newcommand{\xone}{\ensuremath{x_1}}
\newcommand{\semn}[1]{\ensuremath{\llbracket \varphi_{{#1}} \rrbracket}}
\newcommand{\xtwo}{\ensuremath{x_2}}
\newcommand{\tform}[2]{\ensuremath{\tau_{#1}^{#2} }}
\newcommand{\mycolornegb}[1]{{\color{\colornegb}{#1}}}
\newcommand{\getSubset}{GetSubset}
\newcommand{\getSAUNF}{GetSAUNF}

\newcommand{\lits}{\ensuremath{\mathsf{lits}}}
\newcommand{\XX}{\ensuremath{\mseq{X}}}
\newcommand{\II}{\ensuremath{\mseq{I}}}
\newcommand{\YY}{\ensuremath{\mseq{Y}}}

\usepackage{lipsum,atbegshi,etoolbox}
\makeatletter
\newcommand{\discardpages}[1]{
  \xdef\discard@pages{#1}
  \AtBeginShipout{
    \renewcommand*{\do}[1]{
      \ifnum\value{page}=##1\relax%
        \AtBeginShipoutDiscard
        \gdef\do####1{}
      \fi%
    }%
    \expandafter\docsvlist\expandafter{\discard@pages}
  }%
}

\usepackage{pdfpages}
\usepackage{stmaryrd}
\usepackage{lineno}
\AtEndDocument{\par\leavevmode}
\begin{document}
\title{A Normal Form Characterization for Efficient Boolean Skolem Function Synthesis}
\author{\IEEEauthorblockN{Preey Shah, Aman Bansal, S. Akshay and Supratik Chakraborty \thanks{This work was partly supported by
DST/CEFIPRA/INRIA project EQuaVE, SERB Matrices grant MTR/2018/000744 and MHRD IMPRINT-1 grant (Project 6537) of Government of India.}}\IEEEauthorblockA{Department of Computer Science and Engineering\\ Indian Institute of Technology Bombay, Mumbai, India.\\Email: \{preeyshah, aman0456b\}@gmail.com; \{akshayss, supratik\}@cse.iitb.ac.in}
}
\IEEEoverridecommandlockouts
\maketitle
\begin{abstract}
Boolean Skolem function synthesis concerns synthesizing outputs as Boolean functions of inputs such that a relational specification between inputs and outputs is satisfied. This problem, also known as Boolean functional synthesis, has several applications, including design of safe controllers for autonomous systems, certified QBF solving, cryptanalysis etc. Recently, complexity theoretic hardness results have been shown for the problem, although several algorithms proposed in the literature are known to work well in practice. This dichotomy between theoretical hardness and practical efficacy has motivated research on normal forms of specification representation that guarantee efficient synthesis, thus partially explaining the efficacy of some of these algorithms.

In this paper we go one step further and ask if there exists a normal form representation of the specification that precisely characterizes ``efficient" synthesis. We present a normal form called SAUNF that answers this question affirmatively.  Specifically, a specification is polynomial time synthesizable iff it can be compiled to SAUNF in polynomial time. Additionally, a specification admits a polynomial-sized functional solution iff there exists a semantically equivalent polynomial-sized SAUNF representation. SAUNF is exponentially more succinct than well-established normal forms like BDDs and DNNFs, used in the context of AI problems, and strictly subsumes other more recently proposed forms like SynNNF. It enjoys compositional properties that are similar to those of DNNF. Thus, SAUNF provides the right trade-off in knowledge representation for Boolean functional synthesis.
\end{abstract}

\section{Introduction}\label{sec:intro}
The history of Skolem functions can be traced all the way back to the 1920's when Thoralf Skolem provided a simplified proof of the celebrated L\"{o}wenheim Skolem theorem in first order logic. A key step in the proof showed that any first order logic formula can be converted into Skolem normal form, that has no existential quantifiers, while preserving satisfiability. This process, called Skolemization, involves replacing existentially quantified variables by terms constructed out of new function symbols, called Skolem functions.  Skolemization has been an immensely influential technique in logic, and is now used routinely in many applications, viz. automated theorem proving. While it suffices for some applications to merely know that desired Skolem functions exist, others require us to efficiently synthesize such Skolem functions.

Algorithmic synthesis of Skolem functions has been studied extensively in the Boolean setting. Given disjoint sequences of Boolean variables $\mseq{I}= (i_1, \ldots i_n)$ and $\mseq{X}= (x_1, \ldots x_m)$, representing inputs and outputs respectively of a system, 
 and given a Boolean formula $\varphi(\mseq{X}, \mseq{I})$ specifying a desired relation between the system inputs and outputs, the {\em Boolean Skolem function synthesis} ($\bfs$) problem asks us to synthesize a sequence of formulas $\bpsi(\mseq{I})=(\psi_1(\mseq{I}),\ldots,\psi_m(\mseq{I}))$ that can be substituted for $\mseq{X}$ to satisfy the specification, i.e., $\forall \mseq{I}\,\big(\varphi(\bpsi(\mseq{I}), \mseq{I}\big) \Leftrightarrow \exists \mseq{X}\ \varphi(\mseq{X}, \mseq{I})\big)$. The formulas in $\bpsi$ indeed represent Boolean Skolem functions for $\mseq{X}$ in $\exists \mseq{X}\varphi(\mseq{X}, \mseq{I})$\footnote{We are conflating functions and formulas here for simplicity, the distinction will be made clear later.}.

The above problem, also referred to as \emph{Boolean functional synthesis} in the literature, has several applications; we will just mention two here. Skolem functions (and their counterparts, called Herbrand functions) can be thought of as ''certificates" that help us independently verify the results of satisfiability checking for Quantified Boolean Formulas, as done in~\cite{jiang}.
QBF-satisfiability solving is used today in diverse applications, from planning to program repair to reactive synthesis and the like~\cite{survey}.  Having certificates not only helps in verifying correctness of QBF-satisfiability results, but also has other benefits like providing a feasible plan in a planning problem. Yet another application of Skolem functions is motivated by cryptanalysis.
Consider a system with a single $2n$-bit unsigned integer input $\mseq{I}$, and two $n$-bit unsigned integer outputs $\mseq{X}_1$ and $\mseq{X}_2$.
Suppose the relational specification is given as $\varphi_{\mathit{fact}}( \mseq{X}_1, \mseq{X}_2, \mseq{I}) \equiv ((\mseq{I} = \mseq{X}_1\times_{[n]} \mseq{X}_2) \wedge (\mseq{X}_1\neq \mathbf{1}) \wedge (\mseq{X}_2\neq \mathbf{1}))$, where $\times_{[n]}$ denotes $n$-bit unsigned integer multiplication and $\mathbf{1}$ denotes an $n$-bit representation of the integer $1$.
This specification can be represented as a Boolean formula of size $\mathcal{O}(n^2)$ over the variables in  $\mseq{I}$, $\mseq{X}_1$ and $\mseq{X}_2$.
Finding Skolem functions for $\mseq{X}_1$ and $\mseq{X}_2$ in terms of $\mseq{I}$ effectively asks us to solve the ($n$-bit) factorization problem.
Note that if $\mseq{I}$ represents a prime number, there are no values of  $\mseq{X}_1,\mseq{X}_2$ that satisfy the specification.
Hence the specification is technically {\em unrealizable};  yet, it is of significant interest (e.g. in cryptanalysis) to synthesize Skolem functions for $\mseq{X}_1, \mseq{X}_2$ that can be evaluated efficiently.
Note that it is an open question whether there are polynomial time algorithms or even non-uniform polynomial sized circuits for integer factorization. 

Given its significance, the Boolean Skolem function synthesis problem has received considerable attention over the last two decades, with a lot of work focussed towards design of practically efficient algorithms~\cite{Rabe15,bierre,rabe2,jiang2,rsynth,tacas2017,bafsyn:fmcad2018,fmcad2015:skolem,KMPS10,cav18,cav20}.
These algorithms, using techniques ranging from CEGAR to decision tree learning, empirically work well on large collections of benchmarks, but fail remarkably for some small benchmarks.
Further, and somewhat surprisingly, each tool seems to work well on a different set of benchmarks, often incomparable across tools.
What is common among the approaches, however, is that it is difficult to predict reliably the set of benchmarks, i.e.,  the class of formulas, on which a particular algorithm will be efficient (other than simple cases or heuristic guesses).
On a related note, a theoretical study undertaken in~\cite{cav18} showed that Boolean Skolem function synthesis requires super-polynomial space and time unless some well-regarded complexity-theoretic conjectures are falsified.
In fact,~\cite{cav18} also showed that under some weaker assumptions, there cannot exist even sub-exponential algorithms for this problem. 

This leads to a curious dichotomy of theoretical worst-case hardness vs practical (sometimes unreasonable) efficiency.
To resolve this dichotomy, researchers have searched for structure in the input specification that can result in provably efficient synthesis.
It turns out that the representation used for input specification and Skolem functions indeed has a bearing on the complexity of synthesis.
For example, if the specification is given as a {\ROBDD}~\cite{bryant1986} with \emph{input-first variable ordering} (see~\cite{rsynth} for details), there exists a polynomial-time algorithm that generates Skolem functions as {\ROBDD}s~\cite{rsynth}. In~\cite{fmcad19}, a new normal form for specifications, called {\snf}, was proposed, which ensures polynomial-time synthesis, assuming both the specification and Skolem functions are represented as arbitrary Boolean circuits.
However, these earlier studies only provide sufficient but not necessary conditions for efficient Skolem function synthesis.
Significantly, it is not the case that every class of specifications that admit efficient Skolem function synthesis can be efficiently compiled to {\ROBDD}s with input-first variable ordering, or even to {\snf}.
Indeed, \cite{fmcad19} gives (counter-)examples of specifications that are not in {\snf} but admit efficient Skolem function synthesis.  

In this paper, we address the above dichotomy, by presenting  a normal form for Boolean circuits, called \eqnf\  (acronym for {\em Subset-And-Unrealizable Normal Form}), that characterizes polynomial-time and polynomial-sized Boolean Skolem function synthesis.
By a characterization,  we mean that for \emph{every} class $\mathcal{C}$ of circuits (i) Skolem functions can be synthesized in polynomial-time for specifications represented by circuits in $\mathcal{C}$ iff these circuits can be compiled to semantically equivalent ones in \eqnf\ in polynomial-time, and (ii) specifications represented by circuits in $\mathcal{C}$ admit polynomial-sized Skolem functions iff they can be compiled into polynomial-sized semantically equivalent circuits in \eqnf.  This notion
is made precise later in  Section~\ref{sec:prob-stmt}.
We also explore the proposed normal form in depth,  and present several interesting results. Our main contributions are the following. 
\begin{itemize}
\item We show that \eqnf\ is (often exponentially) more succinct and strictly subsumes several other sub-classes (viz. {\DNNF}, {\dDNNF}~\cite{ddnnf}, {\wDNNF}~\cite{cav18}, {\ROBDD}, {\snf}). 
\item We present a polynomial-time algorithm to synthesize polynomial-sized Skolem functions from specifications in $\eqnf$.
\item We study compositional properties of \eqnf\ including disjunction and conjunction operations.
\item We show that checking membership in \eqnf\ is Co-NP hard and is in the second level of the polynomial hierarchy.
\item We present a novel algorithm for compiling a Boolean relational specification in {\CNF} to {\eqnf}.
\end{itemize}
Finally, we show an interesting application of \eqnf.
Specifically, we show that in the context of the $n$-bit factorization problem mentioned earlier, there exist polynomial-sized \eqnf\ circuits relating specific bits of the input $\mseq{I}$ to the outputs $\mseq{X}_1$ and $\mseq{X}_2$.
While this does not solve the $n$-bit factorization problem,  it is worth noting that some of these bit relations are known to require exponentially large {\ROBDD}s~\cite{bryant91}, and sub-exponential sized
circuits using normal forms like {\DNNF}, {\dDNNF}, {\snf} are not yet known.

Normal forms for Boolean functions, or knowledge representation in general, have been investigated extensively over the last few decades~\cite{CD97,dnnf,dsharp,DM11}. While a problem compiled to a normal form may allow the problem to be solved efficiently, compilation to the normal form may not always be easy.
For instance, {\dDNNF} allows polynomial-time model counting, but converting to {\dDNNF} cannot always be done in polynomial time unless ${\myP}=\#{\myP}$.
Despite the worst-case complexity of the compilation process, research in normal forms offers several benefits, such as better understanding of compositionality and other structural properties, explanations for practical performance of algorithms (e.g., on benchmarks in a normal form that permits efficient analysis) etc.
Furthermore, the study of normal forms also feeds into research on normal form compilers, that have significant practical use.
For example, multiple {\dDNNF} compilers have been developed since the introduction of {\dDNNF} as a useful normal form. We also point out that different types of normal forms have been studied earlier. For example, syntactic or purely structural normal forms like {\CNF}, {\DNF}, {\DNNF} allow efficient membership checking, while semantic normal forms like {\dDNNF} require propositional satisfiability checks to determine membership.
The proposed normal form ({\eqnf}) falls in the latter category, but like {\dDNNF}, is worth studying for the good properties it exhibits.
 
 The remainder of this paper is organized as follows.  We start with preliminaries in Section~\ref{sec:prelim} and problem statement in Section~\ref{sec:prob-stmt}. In Section~\ref{sec:saunf}, we introduce \eqnf\ and compare it with other normal forms in Section~\ref{sec:saunf-reln}. In Section~\ref{sec:skolem-synthesis} we explain how Skolem functions can be efficiently computed from \eqnf\ specifications, and discuss compositionality properties in Section~\ref{sec:operations}.
 Next, in Section~\ref{sec:conversion}, we  describe an algorithm to compile a CNF formula to \eqnf. Finally, we show applications to $n$-bit factorization in Section~\ref{sec:appl} and conclude in Section~\ref{sec:concl}. 

 \section{Preliminaries}\label{sec:prelim}
 
Let $\mseq{V} = (v_1, \ldots, v_r)$ be a finite sequence of Boolean variables.  We use $\mset{\mseq{V}}$ to denote the underlying set of the sequence and $|\mseq{V}|$ to denote the length of the sequence.  A \emph{literal} $\ell$ over $\mseq{V}$ is either $v$ or $\neg v$, where $v \in \mset{\mseq{V}}$.  The set of all literals over $\mseq{V}$ is
denoted $\lits(\mseq{V})$. We use $\T$ and $\F$ to represent the Boolean constants $\ttrue$ and $\ffalse$ respectively. A Boolean \emph{formula} $\varphi$ over $\mseq{V}$ is defined by the grammar:
\begin{tabular}{rcl}
$\varphi$ & $::=$ & $\neg \varphi~\vert~ (\varphi) ~\vert~ \varphi \wedge \varphi ~\vert~ \varphi \vee \varphi ~\vert~ \T ~\vert~ \F ~\vert~ v_1 ~\vert \cdots \vert~ v_{r}$ 
\end{tabular}.
We write $\varphi(\mseq{V})$ to denote that the formula $\varphi$ is defined over the sequence of variables $\mseq{V}$.  Special cases of formulas include \emph{clauses} or disjunctions of literals, and \emph{cubes} or conjunctions of literals.  
A formula is in \emph{conjunctive normal form} (\CNF) if it is a conjunction of clauses.
Similarly, it is in \emph{disjunctive normal form} (DNF) if it is a disjunction of cubes. A Boolean \emph{function} $f(\mseq{V})$ is a mapping $\{\F,\T\}^{|\mseq{V}|} \rightarrow \{\F,\T\}$.
The semantics of the Boolean formula $\varphi(\mseq{V})$ is given by a Boolean function $\llbracket \varphi \rrbracket(\mseq{V}): \{\F,\T\}^{|\mseq{V}|} \rightarrow \{\F,\T\}$.  It is easy to see that every Boolean function $f$ corresponds to at least one Boolean formula $\varphi$ such that $\llbracket \varphi \rrbracket = f$.  

Let $\mseq{U}$ be a sub-sequence of $\mseq{V}$ (this includes the possibility $\mseq{U} = \mseq{V}$), and let $\mseq{V}\setminus\mseq{U}$ denote the sequence obtained by removing from $\mseq{V}$ all variables present in $\mseq{U}$.  An \emph{assignment} of $\mseq{U}$ is a mapping $\sigma: \mset{\mseq{U}} \rightarrow \{\F, \T\}$.  We use $\llbracket \varphi \rrbracket_{\sigma}$ to denote the Boolean function $\{\F,\T\}^{|\mseq{V}\setminus \mseq{U}|} \rightarrow \{\F,\T\}$ obtained by substituting $\sigma(v_j)$ for every variable $v_j \in \mset{\mseq{U}}$ in $\llbracket \varphi \rrbracket$.  We say that the formula $\varphi(\mseq{V})$ \emph{reduces} to the formula $\varphi'(\mseq{V}\setminus\mseq{U})$ under the assignment $\sigma$ of $\mseq{U}$ iff $\llbracket \varphi' \rrbracket = \llbracket \varphi \rrbracket_{\sigma}$.  We say that $\sigma$ \emph{satisfies} $\varphi$ if $\llbracket \varphi \rrbracket_\sigma$ always evaluates to $\T$.

We choose to represent both Boolean functions and Boolean formulas (modulo semantic equivalence) by Boolean circuits.  For purposes of this paper, a Boolean circuit (or simply a circuit) is a rooted directed acyclic graph (DAG) $G$ in which nodes with incoming edges, also called \emph{internal nodes}, are labeled by $\lor$, $\land$ and $\neg$ operators, and nodes with no incoming edges, also called \emph{leaves}, are labeled either by variables in $\mseq{V}$ or by constants in $\{\F, \T\}$.  Every internal node labeled $\wedge$ or $\vee$ has incoming edges from exactly two children, while every internal node labeled $\neg$
has an incoming edge from exactly one child.  In order to ensure that a circuit doesn't have superfluous nodes, we require all nodes in a circuit to be descendants of the root. The size of a circuit $G$, denoted $|G|$, is the number of nodes in $G$.  A circuit $G$ represents a Boolean formula $\varphi_G$ (alternatively, a Boolean function $\llbracket \varphi_G \rrbracket$, if $G$ is used to represent a Boolean function) defined as follows:
(i) if $G$ consists of a single leaf labeled $\lambda$, then $\varphi_G = \lambda$; (ii) if the root of $G$ is labeled $\mathsf{op}\in \{\wedge, \vee\}$ and if the sub-circuits rooted at its children are $G_1$ and $G_2$, then $\varphi_G = \varphi_{G_1} ~\mathsf{op}~ \varphi_{G_2}$; (iii) if the root of $G$ is labeled $\neg$ and if the sub-circuit rooted at its (only) child is $H$, then $\varphi_G = \neg \varphi_H$. 

A Boolean formula is said to be in \emph{negation normal form} (or NNF) if the application of $\neg$ is restricted to only the variables.  Motivated by this, a circuit in which every $\neg$ labeled node has a leaf labeled by a variable as its child is said to be an \emph{NNF circuit}.   It is well-known that every Boolean formula is semantically equivalent to a formula in NNF.  Since we wish to reason about Boolean formulas/functions modulo semantic equivalence, it suffices to restrict our attention to NNF circuits. \emph{For convenience of exposition, all circuits in the remainder of the paper are assumed to be in NNF, unless stated otherwise}. It is easy to see that an arbitrary Boolean circuit $G$ can be converted to an NNF circuit $G'$ such that $\llbracket \varphi_G \rrbracket = \llbracket \varphi_{G'} \rrbracket$, and $|\varphi_{G'}| \leq 2\times|\varphi_G|$.  For notational convenience, we treat a $\neg$ labeled node with a child labeled $v$ in a NNF circuit, as a new leaf labeled $\neg v$.  Thus, an NNF circuit can be viewed as a rooted DAG with $\land$- and $\lor$-labeled internal nodes and leaves labeled by literals over $\mseq{V}$.  Note that popular representations of Boolean functions, viz. lists of (implicitly conjoined) clauses, lists of (implicitly disjoined) cubes, and-inverter graphs~\cite{aiger}, ROBDDs~\cite{bryant1986}, DNNF/dDNNF circuits~\cite{dnnf,ddnnf} etc. can all be translated to NNF circuits in linear time.  Figure~\ref{fig:eqNNF} shows an example of an NNF circuit.  In this figure, the annotations $G$, $G_1$ and $G_2$ represent the (sub-)circuits rooted at the nodes adjacent to the annotations.   The leaves are designated $L_0$ through $L_{15}$ from left to right.  As is the case in this figure, multiple leaves of a circuit may have the same literal label.  

Let $L$ be a subset of leaves of circuit $G$. 
We say $L$ is \emph{literal-consistent} in $G$ if every leaf in $L$ is labeled by the same literal. For a literal-consistent set $L$ of leaves in $G$, and for $b \in \{\F,\T\}$, we use $G\mid_{L:b}$ to denote the circuit obtained by re-labeling each leaf of $G$ in the set $L$ with $b$.  For a literal $\ell$ over $\mseq{V}$, we use the term $\ell$-\emph{leaves} of $G$ to denote the set of all leaves of $G$ labeled $\ell$.  For a set of distinct literals $\{\ell_1, \ldots, \ell_r\}$  and (possibly same) labels $b_1, \ldots, b_r$, we abuse notation and use $G\mid_{\ell_1=b_1, \ldots, \ell_r=b_r}$ to denote the circuit obtained by re-labeling  all $\ell_j$-leaves of $G$ by $b_j$, for all $j \in \{1, \ldots, r\}$.  Note that since $\ell$ and $\neg \ell$ are different literals, the notation $G\mid_{\ell=b, \neg\ell=b}$ is meaningful (and useful), and represents the circuit obtained by re-labeling all $\ell$-leaves and $\neg \ell$-leaves of $G$ by $b$.

Let $\mseq{I} = (i_1, \ldots, i_n)$ and $\mseq{X} = (x_1, \ldots, x_m)$ be disjoint sequences of Boolean variables representing inputs and outputs, respectively, of a hypothetical system.  For clarity of exposition, we use "system inputs" to refer to $\mseq{I}$, and "system outputs" to refer to $\mseq{X}$.  Consider a circuit $G$ with leaves labeled by $\lits(\mseq{X})$ and $\lits(\mseq{I})$. The formula $\varphi_G(\mseq{X}, \mseq{I})$ represents a relational specification over the system inputs $\mseq{I}$ and system outputs $\mseq{X}$.  Given $G$, the \emph{Boolean Skolem Function Synthesis} or $\bfs$ problem requires us to find a sequence of Boolean formulas $\bpsi(\mseq{I}) = \big(\psi_1(\mseq{I}), \ldots, \psi_m(\mseq{I})\big)$ such that  $\forall \mseq{I}\,\big(\varphi_G(\bpsi(\mseq{I}), \mseq{I}\big) \Leftrightarrow \exists \mseq{X}\ \varphi_G(\mseq{X}, \mseq{I})\big)$.
As seen earlier, this is an important problem with diverse applications. We call $\psi_j(\mseq{I})$ a \textit{Skolem function}\footnote{Technically, $\llbracket \psi_j \rrbracket$ is the Boolean Skolem function for $x_j$ in $\varphi_G$. However, since we represent both Boolean functions and formulas as circuits, we use $\psi_j$ and $\llbracket \psi_j \rrbracket$ interchangeably for Boolean Skolem functions, to keep the notation simple.} 
for $x_j$ in $\varphi_G(\mseq{X}, \mseq{I})$, and the sequence (or vector) of all such Skolem functions for $x_1, \ldots, x_m$  a \emph{Skolem function vector} for $\mseq{X}$ in $\varphi_G(\mseq{X}, \mseq{I})$.  Since we have chosen to represent all Boolean formulas and functions as circuits, we require each $\psi_j(\mseq{I})$ to be presented as a circuit.
\begin{example}
Let $\mseq{X} = (x_1, x_2)$ and $\mseq{I} = (i)$.  Let $G$ be the circuit shown in Figure~\ref{fig:eqNNF}.  Then $\varphi_G(\mseq{X}, \mseq{I})$ is a relational specification over $\mseq{I}$ and $\mseq{X}$, and one (of possibly many) Skolem function vectors for $\mseq{X}$ in $\varphi_G$ is $\bpsi(\mseq{I}) = \big(\psi_1(\mseq{I}), \psi_2(\mseq{I})\big)$, where
$\psi_1(\mseq{I}) = \neg i = \psi_2(\mseq{I})$. Indeed, it can be verified that $\forall \mseq{I}\, \big(\varphi_G(\bpsi(\mseq{I}), \mseq{I}) \Leftrightarrow \exists \mseq{X} \varphi_G(\mseq{X}, \mseq{I})\big)$.
\end{example}

    \begin{figure*}
        \begin{center}
    \scalebox{0.7}{
    \begin{tikzpicture}[
        edge from parent/.style={draw,latex-},
        var/.style={draw,circle, minimum size=1cm, black},
        op/.style={draw,circle, minimum size=0.75cm, black},
        level distance=1.5cm,
        level 1/.style={sibling distance=10cm},
        level 2/.style={sibling distance=5cm},
        level 3/.style={sibling distance=2.5cm},
        level 4/.style={sibling distance=1.25cm}]
    \node[op,label=right:{~~$G$}] {$\vee$}
    child {node[op,label=left:{$G_1$~~}] {$\wedge$}
        child {node[op] {$\vee$}
            child {node[op] {$\wedge$}
                child {node[var,label=below:{$L_0$}] {$i$}}
                child {node[var,label=below:{$L_1$}] {$\mycolornega{\neg x_1}$}}
            }
            child {node[op] {$\wedge$}
                child {node[var,label=below:{$L_2$}] {$\neg{i}$}}
                child {node[var,label=below:{$L_3$}] {${x_1}$}}
            }
        }
        child {node[op] {$\vee$}
            child {node[op] {$\wedge$}
                child {node[var,label=below:{$L_4$}] {$i$}}
                child {node[var,label=below:{$L_5$}] {$\mycolornegb{\neg{x_2}}$}}
            }
            child {node[op] {$\wedge$}
                child {node[var,label=below:{$L_6$}] {$\neg{i}$}}
                child {node[var,label=below:{$L_7$}] {$x_2$}}
            }
        }
    }
    child {node[op,label=right:{~~$G_2$}] {$\wedge$}
        child {node[op] {$\vee$}
            child {node[op] {$\wedge$}
                child {node[var,label=below:{$L_8$}] {$i$}}
                child {node[var,label=below:{$L_9$}] {$\neg{x_2}$}}
            }
            child {node[op] {$\wedge$}
                child {node[var,label=below:{$L_{10}$}]
                {\mycolora{$x_1$}}}
                child {node[var,label=below:{$L_{11}$}] {$\neg x_2$}}
            }
        }
        child {node[op] {$\vee$}
            child {node[op] {$\wedge$}
                child {node[var,label=below:{$L_{12}$}] {$i$}}
                child {node[var,label=below:{$L_{13}$}] {\mycolorb{$x_2$}}}
            }
            child {node[op] {$\wedge$}
                child {node[var,label=below:{$L_{14}$}] {$\neg{x_1}$}}
                child {node[var,label=below:{$L_{15}$}] {$\neg{x_2}$}}
            }
        }
    };
    \end{tikzpicture}
    }
    \end{center}
\caption{Example of an NNF circuit $G$, which is in {\eqnf} w.r.t. $\mset{\mseq{X}} = \{x_1, x_2\}$ and $(\{L_3\}, \{L_7\}, \{L_5\}, \{L_1\})$ }
        \label{fig:eqNNF}
    \end{figure*}

\section{Problem Statement}\label{sec:prob-stmt}
Earlier work~\cite{cav18} has established (conditional) time and space lower bounds for $\bfs$; therefore it is unlikely that efficient algorithms exist for solving this problem in general.  Yet, several recent works~\cite{cav18,rsynth,rabe2,rabe3,cav20,fmsd20} have shown that $\bfs$ indeed admits practically efficient solutions for several non-trivial benchmarks.  This motivates us to ask the following question, where we are interested in Boolean circuit representations of relational specifications and Skolem functions.

\noindent \emph{Does there exist a class, say $\mathcal{C}^\star$, of circuits such that the following hold?
\begin{enumerate}
    \item[P0:] For every circuit $G$, there is a semantically equivalent circuit $G^\star \in \mathcal{C}^\star$, i.e. $\llbracket \varphi_G \rrbracket = \llbracket \varphi_{G^\star} \rrbracket$. In other words, $\mathcal{C}^\star$ is not semantically constraining.
    \item[P1:] $\bfs$ is solvable in polynomial-time for the class $\mathcal{C}^\star$.
    \item[P2:] For every class $\mathcal{C}$ of circuits, 
    \begin{itemize}
        \item[P2a:] $\bfs$ is solvable in polynomial-time for the class $\mathcal{C}$ iff circuits in $\mathcal{C}$ can be compiled to semantically equivalent ones in $\mathcal{C}^\star$ in polynomial-time.
        \item[P2b:] Relational specifications represented by circuits in $\mathcal{C}$ admit polynomial-sized Skolem function vectors iff circuits in $\mathcal{C}$ admit polynomial-sized semantically equivalent circuits in $\mathcal{C}^\star$.
    \end{itemize}
\end{enumerate}
}

We answer the above question positively in this paper, effectively providing a circuit normal form characterization of efficient Boolean Skolem function synthesis.    In light of our characterization, the hardness results of~\cite{cav18} translate to the hardness of computing $G^\star \in \mathcal{C}^\star$  such that $\llbracket \varphi_{G^\star} \rrbracket = \llbracket \varphi_G \rrbracket$.

\section{A Normal Form for Synthesis}\label{sec:saunf}
Let $G$ be a circuit with leaves labeled by  $\lits(\mseq{I})$ and $\lits(\mseq{X})$.  Let $\ell$ be a literal labeling a leaf of $G$, and let $v_\ell$ be the underlying variable of $\ell$.  \emph{Throughout this section, we assume that $w, w'$ are fresh variables not in $\mseq{I}$ or $\mseq{X}$}.  
\begin{definition}\label{def:l-unreal}
We say that $\ell$ is \emph{$\land$-realizable} in $G$ iff there is an assignment $\sigma: (\mset{\mseq{I}} \cup \mset{\mseq{X}})\setminus \{v_\ell\} \rightarrow \{\F,\T\}$ such that $\llbracket G\mid_{\ell= w, \neg \ell = w'} \rrbracket_{\sigma} = \llbracket (w \wedge w') \rrbracket$. Furthermore, we say that $\ell$ is \emph{$\land$-unrealizable} in $G$ iff it is not $\land$-realizable in $G$. 
\end{definition}
Intuitively, $\ell$ is $\land$-realizable in $G$ if $\varphi_G$ reduces to $w \wedge w'$ under some assignment of variables other than $v_{\ell}$, after $\ell$ and $\neg \ell$ are replaced by $w$ and $w'$ respectively in the leaves of $G$.  It is easy to see that \emph{if $\ell$ is $\land$-realizable (resp. $\land$-unrealizable) in $G$, then so is $\neg \ell$}.
\begin{example}
Consider circuit $G$ in Figure~\ref{fig:eqNNF}, and let $G_1$ and $G_2$ denote the sub-circuits rooted at the left and right child, respectively of the root node. Then $x_1$ is $\land$-realizable in $G_2$ and $G$ (use $\sigma(i) = \sigma(x_2) = \F$) but is $\land$-unrealizable in $G_{1}$.
\end{example}
We now extend the notion of $\land$-(un)realizability to that of sets of literal-consistent leaves.  If $S$ is the set of all $\ell$-leaves of $G$, the notion of $S$ being $\land$-realizable (resp. $\land$-unrealizable) in $G$ naturally coincides with that of literal $\ell$ being $\land$-realizable (resp. $\land$-unrealizable) in $G$.  However, if $S$ does not contain all $\ell$-leaves of $G$, we must specify what to do with leaves labeled $\ell$ but not in $S$. The following definition does exactly that.
\begin{definition}\label{def:set-realizability}
Let $S$ be a literal-consistent set of leaves of $G$, and let $\ell$ be the literal labeling each leaf in $S$. Let $S'$ be the set of all $\ell$-leaves of $G$.
We say that $S$ is $\land$-realizable (resp. $\land$-unrealizable) in $G$ if $\ell$ is $\land$-realizable (resp. $\land$-unrealizable) in $G\mid_{S'\setminus S: \F}$ 
\end{definition}
Thus, all $\ell$-leaves that are not in $S$ must be labeled $\F$ before we check whether $\ell$ is $\land$-realizable in the resulting circuit.
\begin{example}
Referring back to Figure~\ref{fig:eqNNF}, we wish to check the $\land$-(un)realizability of $S = \{L_3\}$ in $G$.  The literal labeling $L_3$ is $x_1$ and the set of all $x_1$-leaves is $S' = \{L_3, L_{10}\}$.  Hence $S' \setminus S = \{L_{10}\}$.  To check the $\land$-(un)realizability of $S$, we re-label $L_{10}$ with $\F$, and $L_3$ with a fresh variable $w$.  Additionally, all leaves labeled $\neg x_1$, i.e. $L_1$ and $L_{14}$ are re-labeled with a fresh variable $w'$.  

Let $G'$ denote the resulting circuit.  We now ask if there is an assignment $\sigma: \{i, x_2\} \rightarrow \{\F,\T\}$ such that $\llbracket \varphi_{G'} \rrbracket_\sigma = \llbracket w \wedge w' \rrbracket$.  From the circuit structure of $G'$, we can see that there is only one leaf, viz. $L_3$, labeled $w$.  Hence, in order to have $\llbracket \varphi_{G'} \rrbracket_\sigma = \llbracket w \wedge w' \rrbracket$, the assignment $\sigma$ must not mask the value of $L_3$ from ``propagating" up to the root of $G'$.  This implies that $L_2$ must be labeled $\T$, i.e. $\sigma(i) = \F$, and $\sigma(x_2) = \neg \sigma(i) = \T$.  With this $\sigma$, it is now easy to verify that  $\llbracket \varphi_{G'} \rrbracket_\sigma = \llbracket w \rrbracket \neq \llbracket w \wedge w' \rrbracket$.  Hence, there is no assignment of $x_2$ and $i$ that renders the formula represented by $G'$ semantically equivalent to $w \land w'$.  It follows that $S = \{L_3\}$ is $\land$-unrealizable in the circuit shown in  Figure~\ref{fig:eqNNF}.
A similar exercise shows that $\widehat{S} = \{L_{10}\}$ is $\land$-realizable in the same circuit (use $\sigma(i) = \sigma(x_2) = \F$).
\end{example}

Finally, we use the above definitions to introduce a new normal form for circuits that precisely characterizes efficient Boolean Skolem Function Synthesis.  We show in subsequent sections that this normal form defines a class $\mathcal{C}^\star$ of circuits that satisfies properties P0, P1 and P2 described in Section~\ref{sec:prob-stmt}.
\begin{definition}\label{def:saunf}
Let $G$ be a circuit with leaves labeled by  $\lits(\mseq{I})$ and $\lits(\mseq{X})$. Let $S = (S_1, S_2,.. S_k)$ be a non-empty sequence of subsets of leaves of $G$. We say that $G$ is in \emph{Subset And-Unrealizable Normal Form  ({\eqnf}, for short)} w.r.t. $\mset{\mseq{X}}$ and $S$ if the following hold: 
\begin{enumerate}
    \item $S_j \cap S_l = \emptyset$ for all distinct $j, l \in \{1, \ldots k\}$.
    \item For each $j \in \{1, \ldots k\}$,  all leaves in $S_j$ are labeled by the same literal over $\mseq{X}$.
    \item $S_1$ is $\land$-unrealizable in $G$.
    \item For each $j \in \{2 \ldots k\}$, $S_j$ is $\land$-unrealizable in $G\mid_{S_1:\T, S_2:\T \ldots S_{j-1}:\T}$.
    \item $\llbracket \varphi_{G\mid_{S_1:\T,S_2:\T\ldots S_k:\T}} \rrbracket$ is semantically independent of $\mseq{X}$, i.e. its value doesn't depend on the assignment of $\mseq{X}$.
\end{enumerate} 
\end{definition}
A few points about Definition~\ref{def:saunf} are worth noting.  
\begin{itemize}
    \item A circuit $G$ may be in {\eqnf} w.r.t. $\mset{\mseq{X}}$ and $S$, but not in {\eqnf} w.r.t. a different $\mset{\mseq{X}'}$ and/or $S'$. 
    \item Conditions $3$, $4$ and $5$ are semantic in nature.  Normal forms with such semantic conditions are not new.  For example, the widely used disjoint decomposable negation normal form ({\dDNNF}) uses a semantic condition in its definition (see ~\cite{ddnnf}). 
    \item $S_1 \cup \cdots S_k$ may not include all leaves of $G$, nor even all leaves labeled by a literal over $\mseq{X}$.
       \item While the use of $\T$ as labels for leaves in $S_1, S_2, \ldots$ in conditions 4 and 5 may seem arbitrary for now, we will soon see the significance of this in the synthesis of Boolean Skolem functions.
\end{itemize}

\begin{example}
Consider the circuit $G$ in Figure \ref{fig:eqNNF} again, with $\mseq{I} = (i)$ and  $\mseq{X} =(x_1, x_2)$.  Let $S = (\{L_3\}, \{L_7\}, \{L_5\}, \{L_1\})$ be a sequence of (singleton) subsets of leaves. As seen above, $\{L_3\}$ is $\land$-unrealizable in $G$.  It can similarly be verified that $\{L_7\}$ is $\land$-unrealizable in $G\mid_{\{L_3\}:\T}$, $\{L_5\}$ is $\land$-unrealizable in $G\mid_{\{L_3\}:\T,\{L_7\}:\T}$ and $\{L_1\}$ is $\land$-unrealizable in $G\mid_{\{L_3\}:\T,\{L_7\}:\T, \{L_5\}:\T}$.  Finally, the function represented by $G\mid_{\{L_3\}:\T,\{L_7\}:\T, \{L_5\}:\T, \{L_1\}:\T}$ is semantically equivalent to $\T$, and hence is independent of $\mseq{X}$.  Therefore, the circuit $G$ is in {\eqnf} w.r.t. $\mset{\mseq{X}} = \{x_1, x_2\}$ and $S = (\{L_3\}, \{L_7\}, \{L_5\}, \{L_1\})$.  However, $G$ is not in {\eqnf} w.r.t. $\{x_1, x_2\}$ and $S' = (\{L_{10}\}, \{L_7\}, \{L_5\}, \{L_1\})$, since we have seen earlier that $\{L_{10}\}$ is $\land$-realizable in $G$.
\end{example}

\section{Relation with other normal forms}\label{sec:saunf-reln}
Several normal forms for Boolean circuits studied in the literature, viz. {\ROBDD}~\cite{bryant1986}, {\FDD}~\cite{bryant91}, {\DNNF}~\cite{dnnf}, {\dDNNF}~\cite{ddnnf}, {\wDNNF}~\cite{cav18}, {\snf}~\cite{fmcad19}, admit efficient Boolean Skolem function synthesis, and satisfy properties P0 and P1 in our problem statement (see Section~\ref{sec:prob-stmt}).
However, \emph{none of these are known to satisfy property P2 in our problem statement, thereby failing to provide a characterization of efficient Boolean Skolem function synthesis}. 
In contrast, as we show in this paper, the class of {\eqnf} circuits satisfies all the properties mentioned in our problem statement.  

Among the various alternative normal forms, we discuss {\snf}~\cite{fmcad19} first.
We say that a circuit normal form (or class of circuits) $\mathcal{N}_1$ is exponentially more succinct than another normal form $\mathcal{N}_2$ if 
(i) for every circuit $G_2 \in \mathcal{N}_2$, there exists a circuit $G_1 \in \mathcal{N}_1$ such that $|G_1| \leq |G_2|$ and $\llbracket \varphi_{G_1} \rrbracket ~=~ \llbracket \varphi_{G_2} \rrbracket$, and (ii) there is a circuit $G_1 \in \mathcal{N}_1$ such that every circuit $G_2 \in \mathcal{N}_2$ with $\llbracket \varphi_{G_1} \rrbracket ~=~ \llbracket \varphi_{G_2} \rrbracket$ has $|G_2| \in 2^{\bigO{|G_1|}}$.
The notion of super-polynomial succinctness is similarly defined.
The authors of~\cite{fmcad19} showed a conditional succinctness result for {\snf}, namely {\snf} is super-polynomially more succinct than {\DNNF}~\cite{dnnf} and {\dDNNF}~\cite{ddnnf}, unless some long-standing complexity theoretic conjectures are falsified.
We show the following stronger result for {\eqnf}.
\begin{lemma}\label{lemma:exp-sep}
{\eqnf} is unconditionally exponentially more succinct than {\DNNF} and {\dDNNF}. 
\end{lemma}
\begin{proof}
We use a result from~\cite{Bova16} to prove the lemma. In Proposition 11 of \cite{Bova16}, a family of Boolean functions $\{JS_r \mid r \ge 2\}$ is defined. The formula  $JS_r$, defined on $\bigO{r^2}$ variables, asserts that for every triple $(v_j, v_k, v_l)$ of variables in a carefully constructed set $A_r$ of triples, at least one of $v_j, v_k$ or $v_l$ must be $\ffalse$.  It is shown in \cite{Bova16} that $|A_r| \in \bigO{r^2}$. Therefore, a \CNF formula representing $JS_r$ has $\bigO{r^2}$ clauses, with each clause having three negated variables as literals.  Since no non-negated variables appear as literals in the formula, a circuit representation of the \CNF\ formula cannot have any literal-consistent subset of leaves that is $\land$-realizable. This implies that $JS_r$ can be represented in {\eqnf} in size $\bigO{r^2}$.  It is also shown in~\cite{Bova16} that any {\DNNF} (and hence also {\dDNNF}) representation of $JS_r$ requires size $2^{\Omega\big(r^2\big)}$.  Therefore, {\eqnf} is unconditionally exponentially more succinct compared to {\DNNF} and {\dDNNF}.
\end{proof}
Next, we show that {\snf} is, in fact, a special case of {\eqnf}.  Towards this end, we recall the definition of {\snf} from~\cite{fmcad19}, re-cast in our terminology.
\begin{definition}\label{def:synnnf}
A circuit $G$ with leaves labeled by $\lits(\mseq{I})$ and $\lits(\mseq{X})$ is in {\snf} w.r.t. $\mseq{X}$ iff the following hold:
\begin{itemize}
    \item $x_1$ is $\land$-unrealizable in $G$.
    \item For $2 \le i \le |\mseq{X}|$, $x_i$ is $\land$-unrealizable in $G\mid_{x_1= \T, \neg{x_1} = \T, \ldots x_{i-1} = \T, \neg x_{i-1}= \T}$.
\end{itemize}
\end{definition}
The following lemma shows that {\eqnf} strictly subsumes {\snf}.
\begin{lemma}\label{lem:subsume}
Every circuit $G$ that is in {\snf} w.r.t. $\mseq{X}$ is also a {\eqnf} circuit 
w.r.t $\mset{\mseq{X}}$ and a sequence $S$ of $2\cdot|\mseq{X}|$ $\land$-unrealizable subsets of leaves. However, there exist {\eqnf} circuits that are not in {\snf}.
\end{lemma}
\begin{proof}
Suppose a circuit $G$ is in {\snf} w.r.t. $\mseq{X}$, and let $|\mseq{X}| = r$.  We define a sequence of $2r$ literal-consistent subsets of leaves of $G$ as follows.  For each $j \in \{1, \ldots r\}$, we define $S_{2j-1}$ to be the set of $x_j$-leaves of $G$, and $S_{2j}$ to be the set $\neg x_j$-leaves of $G$.  It can now be seen from Definition~\ref{def:synnnf} and Definition~\ref{def:saunf} that $G$ is in {\eqnf} w.r.t. $\mset{\mseq{X}}$ and the sequence $S = (S_1, \ldots S_{2n})$ of subsets of literal-consistent leaves. This proves the first part of the lemma.

To show the second part, we must demonstrate a circuit that is in {\eqnf} but not in {\snf} for any permutation of the sequence of system outputs $\mseq{X}$.  We claim  that the circuit $G$ in Figure~\ref{fig:eqNNF}, already shown to be in {\eqnf}, suffices for this purpose.
This is because Definition~\ref{def:synnnf} entails that for $G$ to be in {\snf}, at least one literal over $\mseq{X}$ must be $\land$-unrealizable in $G$.  However, none of $x_1, \neg x_1, x_2, \neg x_2$ are $\land$-unrealizable in the circuit $G$ in Figure~\ref{fig:eqNNF}. Specifically, $\llbracket \varphi_{G\mid_{x_1=w,\neg x_1=w'}}  \rrbracket_\sigma = \llbracket w\wedge w' \rrbracket$ when $\sigma(i)= \sigma(x_2) = \F$, and $\llbracket \varphi_{G\mid_{x_2=w,\neg x_2=w'}} \rrbracket_{\sigma'} = \llbracket w\wedge w' \rrbracket$ when $\sigma'(i)= \sigma'(x_1) = \T$.  Therefore, the circuit in Figure~\ref{fig:eqNNF} is not in {\snf} w.r.t. any
permutation of $\mseq{X}$.
\end{proof}
It has been shown in~\cite{fmcad19} that every {\DNNF}, {\dDNNF} and {\wDNNF} circuit is also in {\snf}, and {\snf} is super-polynomially more succinct than {\wDNNF} unless ${\myP} = {\myNP}$.  Furthermore, every {\ROBDD} and {\FDD} can be converted to a {\DNNF} (hence {\snf}) circuit with at most linear blowup in size, although {\snf} can be exponentially more succinct than {\ROBDD} or {\FDD}~\cite{dnnf,fmcad19}.  By virtue of Lemma~\ref{lem:subsume}, we now have the following result.
\begin{corollary}\label{cor:subsume}
All subsumption and (conditional) succinctness results for {\snf} circuits hold for {\eqnf} circuits as well.
\end{corollary}
Since every Boolean specification can be represented as a {\snf} circuit~\cite{fmcad19}, it also follows from Lemma~\ref{lem:subsume} that property P0 in our problem statement (see Section~\ref{sec:prob-stmt}) holds for the class of {\eqnf} circuits.


\section{Efficient synthesis of Skolem functions from {\eqnfs} specifications}\label{sec:skolem-synthesis}
We now show how a Skolem function vector can be efficiently computed if the relational specification is given as a {\eqnf} circuit.  Informally, the process involves transforming a given {\eqnf} circuit $G$ with leaves labeled by $\lits(\mseq{I})$ and $\lits(\mseq{X})$ to a semantically different but related circuit $H$ with leaves labeled by $\lits(\mseq{I})$, $\lits(\mseq{X})$ and $\lits(\mseq{X}')$, where $\mseq{X}'$ is a sequence of fresh system outputs, also called \emph{auxiliary outputs}.  The transformation is done in a way such that a Skolem function vector for $(\mseq{X}, \mseq{X}')$ in $H$ can be found efficiently, and a projection of this Skolem function vector on the first $|\mseq{X}|$ components directly yields a Skolem function vector for $\mseq{X}$ in $G$. To formalize this notion, we begin with a few definitions.

\begin{definition}\label{def:equisynthproj}
[Equisynthesizable Under Projection] Let $G$ be a circuit representing a relational specification over system inputs $\mseq{I}$ and systems outputs $\mseq{X}$.  Let $H$ be another circuit representing a relational specification over $\mseq{I}$ and $(\mseq{X},\mseq{X}')$, where $\mseq{X}'$ is a fresh sequence of system outputs or auxiliary outputs. We say that $G$ is equisynthesizable to $H$ under projection, denoted $\Eup{G}{H}$, iff the following hold
\begin{itemize}
\item $\forall \mseq{I} \forall \mseq{X} \,\big(\varphi_G(\mseq{X}, \mseq{I}) ~\Rightarrow~ \exists \mseq{X}'\, \varphi_H(\mseq{X}, \mseq{X}', \mseq{I})\big)$
\item $\forall \mseq{I} \forall \mseq{X} \forall \mseq{X}' \big(\varphi_H(\mseq{X},\mseq{X}',\mseq{I}) ~\Rightarrow~ \varphi_G(\mseq{X},\mseq{I})\big)$
\end{itemize}
\end{definition}
It follows from Definition~\ref{def:equisynthproj} that \Eup{}{} defines a transitive relation on circuits representing relational specifications.   The following lemma is an easy consequence of Definition~\ref{def:equisynthproj}.
\begin{lemma}
If $\Eup{G}{H}$ holds and  $\big(\bpsi(\mseq{I}), \bpsi'(\mseq{I})\big)$
is a Skolem function vector for $(\mseq{X}, \mseq{X}')$ in $\varphi_H(\mseq{X}, \mseq{X}', \mseq{I})$, then $\bpsi(\mseq{I})$ is a Skolem
function vector for $\mseq{X}$ in $\varphi_G(\mseq{X}, \mseq{I})$.
\end{lemma}
\begin{proof}
 Since $\forall \mseq{I}\,\big(\varphi_G(\bpsi(\mseq{I}), \mseq{I}) \Rightarrow \exists \mseq{X}\varphi_G(\mseq{X}, \mseq{I})\big)$ holds trivially, we only show below
 that $\forall \mseq{I}\,\big(\exists \mseq{X}\varphi_G(\mseq{X}, \mseq{I}) \Rightarrow \varphi_G(\bpsi(\mseq{I}), \mseq{I})\big)$.
 
 From Definition~\ref{def:equisynthproj} and from the definition of Skolem functions, we have
 $\forall \mseq{I}\,\big(\exists \mseq{X} \varphi_G(\mseq{X}, \mseq{I}) \Rightarrow$
 $\exists \mseq{X} \exists \mseq{X}' \varphi_H(\mseq{X}, \mseq{X}', \mseq{I}) \Rightarrow$
 $\varphi_H(\bpsi(\mseq{I}), \bpsi'(\mseq{I}), \mseq{I}) \Rightarrow$
 $\varphi_G(\bpsi(\mseq{I}), \mseq{I})$.
\end{proof}

\subsection{Role of auxiliary outputs}
 We now investigate how auxiliary outputs can be introduced in a principled manner, so that they help in generating Skolem functions. We start with a circuit $G$ with leaves labeled by $\lits(\mseq{I})$ and $\lits(\mseq{X})$. Let $G_1$ and $G_2$ be two sub-circuits of $G$ such that $G_1$ is not a sub-circuit of $G_2$ and vice versa.  For a fresh auxiliary variable $p \not\in \mset{\mseq{X}} \cup \mset{\mseq{I}}$ and for $j \in \{1, 2\}$, define a circuit transformation \tform{j}{p}  that replaces the sub-circuit $G_j$ in $G$ with the circuit representing $\varphi_{G_j} \wedge p$.  The definition of the circuit transformation $\tform{j}{\neg p}$ is similar.
\begin{lemma}\label{onedummyaddition}
 If $\varphi_{G_1} \wedge \varphi_{G_2}$ is unsatisfiable, then \Eup{G}{\tform{1}{p}(\tform{2}{\neg{p}}(G))}.
\end{lemma}
\begin{proof}
 Let $H$ denote the circuit $\tform{1}{p}(\tform{2}{\neg{p}}(G))$, and let $H_1$ and $H_2$ denote the newly introduced sub-circuits representing $p \wedge \varphi_{G_1}$ and $\neg p \wedge \varphi_{G_2}$ respectively in $H$.  We show below that the conditions for $\Eup{G}{H}$ (see Definition~\ref{def:equisynthproj}) are satisfied.
 
 Let $\sigma: \mset{\mseq{X}} \cup \mset{\mseq{I}} \rightarrow \{\F,\T\}$ be an assignment for which $\llbracket \varphi_G \rrbracket$ evaluates to $\T$. We have four cases to analyze depending on what $\semn{G_1}$ and $\semn{G_2}$ evaluate to under $\sigma$.
 \begin{itemize}
    \item $\semn{G_1} = \F = \semn{G_2}$: Then for any assignment to $p$, $\semn{H_1}$ and $\semn{H_2}$ also evaluate to $\F$, and hence $\semn{H}$ evaluates to the same value, viz. $\T$, as $\semn{G}$.
    \item If $\semn{G_1} = \T, \semn{G_2} = \F$, then with $p$ assigned $\T$, $\semn{H_1} = \T$ and $\semn{H_2} = \F$, and hence $\semn{H}$ evaluates to the same value, viz. $\T$, as $\semn{G}$.
    \item By a similar argument, if $\semn{G_1} = \F, \semn{G_2} = \T$, assigning $p$ to $\F$
    causes $\semn{H}$ to evaluate to $\T$.
    \item The case of $\semn{G_1} = \semn{G_2} = \T$ doesn't arise since $\varphi_{G_1} \wedge \varphi_{G_2}$ is unsatisfiable.
\end{itemize}
This shows that $\forall \mseq{I} \forall \mseq{X} \,\big(\varphi_G(\mseq{X}, \mseq{I}) \Rightarrow \exists p\, \varphi_H(\mseq{X}, p, \mseq{I})\big)$.

Consider any assignment $\sigma': \mset{\mseq{X}} \cup \{p\} \cup \mset{\mseq{I}} \rightarrow \{\F,\T\}$ that renders $\semn{H} = \T$.
 Let $\sigma: \mset{\mseq{X}} \cup \mset{\mseq{I}} \rightarrow \{\F,\T\}$ be the projection of $\sigma'$ on $\mset{\mseq{X}} \cup \mset{\mseq{I}}$.  Note that $\sigma'$ necessarily assigns one of $p$ or $\neg p$ to $\F$, while still rendering $\semn{H} = \T$. Therefore, since all internal gates in $H$ (i.e. $\land$ and $\lor$ gates) are monotone, $\sigma$ must render $\semn{H\mid_{p=\T,\neg p =\T}} = \T$ as well.  However, $\semn{H\mid_{p=\T,\neg p =\T}} = \semn{G}$ by definition.  Hence, $\sigma$ satisfies $\semn{G}(\mseq{X}, \mseq{I})$.  This shows that $\forall \mseq{I} \forall \mseq{X} \forall p\, \big(\varphi_H(\mseq{X}, p, \mseq{I}) \Rightarrow \varphi_G(\mseq{X},\mseq{I})\big)$.
 \end{proof}
The argument in the above proof can be easily generalized to prove the following.
\begin{lemma}\label{lemma:sharing}
Let $\mathcal{G}_1 = \{G_{1,1}, \ldots G_{1,s}\}$ and $\mathcal{G}_2 = \{G_{2,1}, \ldots G_{2,t}\}$ be two sets of
sub-circuits of $G$ such that (a) there are no distinct $G_{k,i}$ and $G_{l,j}$ where one is a sub-circuit of the other, and (b) $\bigvee_{i=1}^s \varphi_{G_{1,i}} \Rightarrow \bigwedge_{j=1}^t \neg \varphi_{G_{2,j}}$.  Let $\tform{\mathcal{G}_k}{p}$ (resp. $\tform{\mathcal{G}_k}{\neg p}$) denote the circuit transformation that replaces every sub-circuit $G_{k,i} \in \mathcal{G}_k$ with a subcircuit representing $\varphi_{G_{k,i}} \wedge p$ (resp. $\varphi_{G_{k,i}} \wedge \neg p$), where $p$ is a fresh variable.
Then $\Eup{G}{\tform{\mathcal{G}_1}{p}(\tform{\mathcal{G}_2}{\neg p}(G))}$. 
\end{lemma}
A particularly easy application of Lemma~\ref{lemma:sharing} is
obtained by choosing any literal $\ell$ that labels leaves of $G$,
and by choosing $\mathcal{G}_1$ and $\mathcal{G}_2$ to be subsets of $\ell$-leaves and $\neg\ell$-leaves, respectively.
Note that if $L$ is a subset of $\ell$-leaves of $G$, then $\tform{L}{p}$ gives the same circuit as $G\mid_{L:p\wedge\ell}$.

For the following theorem, consider a relational specification over $\mseq{I}$ and $\mseq{X}$ specified by a circuit $G$. Let $\ell$ be a literal over $\mseq{X}$, $v_\ell$ be the
underlying variable in $\mseq{X}$, and let $S_\ell$ be the set of all $\ell$-leaves of $G$.
\begin{theorem}\label{main} 
Suppose $S \subseteq S_\ell$ is $\land$-unrealizable in $G$.  For a fresh auxiliary variable $p \not\in \mset{\mseq{X}} \cup \mset{\mseq{I}}$, let $E$ denote the circuit  $G\mid_{S_\ell\setminus S:(p \wedge \ell),\, S_{\neg \ell}:(\neg p \wedge \neg \ell)}$ and let $H$ denote the circuit $E\mid_{\ell=\T,\neg \ell=\T}$. Note that the literals labeling leaves of $G$ are from $\lits(\mseq{X})$ and $\lits(\mseq{I})$, those labeling leaves of $E$ are from $\lits(\mseq{X})$, $\lits(\mseq{I})$ and $\{p, \neg p\}$, while the literals labeling leaves of H are from $\lits(\mseq{X} \setminus (v_\ell))$, $\lits(\mseq{I})$ and $\{p, \neg p\}$. Then the following statements hold.
\begin{enumerate}
    \item $\exists v_\ell\, \varphi_G \iff \exists p\, \varphi_H \iff  \exists v_\ell\, \varphi_{G\mid_{S:\T}}$
    \item If $\bpsi_H(\mseq{I})$ is a Skolem function vector for $(\mseq{X}\setminus (v_\ell), p)$ in $H$, then the projection of $\bpsi_H$ on $\mseq{X} \setminus (v_\ell)$ augmented with the Skolem function $\varphi_{E\mid_{\ell=\T,\neg \ell=\F}}(\bpsi_H(\mseq{I}), \mseq{I})$ for $\ell$ gives a Skolem function vector $\bpsi_G(\mseq{I})$ for $\mseq{X}$ in $G$.
\end{enumerate}
\end{theorem}
\begin{proof}
 By Lemma~\ref{lemma:sharing}, we have $\Eup{G}{E}$.  It then follows from Definition~\ref{def:equisynthproj} that 
$\exists v_{\ell}\, \varphi_G \iff \exists v_{\ell} \exists p\, \varphi_E$.
 Furthermore, since $S$ is $\land$-unrealizable in $G$, it follows from the definition of $E$ and from Definition~\ref{def:set-realizability} that $\ell$ is $\land$-unrealizable in $E$.
 From this, we will now show that $\exists v_\ell \varphi_E$ is equivalent to $\varphi_{{E|}_{\ell=\T,\neg \ell=\T}}$. In one direction, we observe that $\exists v_{\ell} \varphi$ is always equivalent to $\varphi_{{E|}_{\ell=\T,\neg \ell=\F}} \vee \varphi_{{E|}_{\ell=\F,\neg \ell=\T}}$. This, in turn, logically implies $\varphi_{{E|}_{\ell=\T,\neg \ell=\T}}$ as all internal gates in an NNF circuit are monotone. In the other direction, $\varphi_{{E|}_{\ell=\T,\neg \ell=\T}} \wedge \neg (\varphi_{{E|}_{\ell=\T,\neg \ell=\F}} \vee \varphi_{{E|}_{\ell=\bot,\neg \ell=\top}})$ is unsatisfiable if $\ell$ is $\land$-unrealizable in $E$ (follows from definition of $\land$-unrealizability).
 Therefore, we have $\exists v_\ell\, \varphi_E\iff \varphi_{E\mid_{\ell=\T,\neg \ell=\T}} \iff \varphi_H$.  Hence, $\exists v_\ell \exists p\, \varphi_E \iff \exists p\, \varphi_H$.  Finally, from the definitions of circuits $E$ and $H$, we have $\exists p\, \varphi_H \iff \exists p\, \varphi_{E|_{l = \T, \neg l = \T}}$ $\iff$ $\exists p\, \varphi_{G|_{S:\T,\, S_\ell\setminus S:p,\,S_{\neg l}:\neg p}}$. By renaming $p$ to $v_l$ in the last formula, we get $\exists v_l\, \varphi_{G|_{S:\T}}$.

Since $\exists v_\ell\, \varphi_E\iff \varphi_H$, a Skolem function vector for $(\mseq{X}\setminus (v_\ell), p)$ in $E$ is obtained from $\bpsi_H (\mseq{I})$.  The Skolem function for $\ell$ in $E$ is then given by $\varphi_{E|_{\ell=\T,\neg \ell=\F}}(\bpsi_H(\mseq{I}), \mseq{I}))$. To see why this works, note that $\varphi_{E}$ with $\bpsi_H(\mseq{I})$ substituted for $(\mseq{X}\setminus (v_\ell), p)$ represents a specification with a single system output $\ell$ and system inputs $\mseq{I}$. Let us call this $\phi(\ell, \mseq{I})$. 
Then, $\phi(\top,\mseq{I})$ serves as a Skolem function, say $\psi^\ell(\mseq{I})$, for ${\ell}$ in $\phi$, i.e., $\exists \ell\, \phi(\ell,\mseq{I})\iff\phi(\psi^\ell(\mseq{I}),\mseq{I})$. Indeed, suppose for some $\mseq{I}$, $\psi^\ell(\mseq{I})= \phi(\top,\mseq{I})= \top$. Then $\phi(\psi^\ell(\mseq{I}),\mseq{I})= \phi(\top,\mseq{I})= \top$. Conversely, if $\psi^\ell(\mseq{I})= \phi(\top,\mseq{I})= \bot$, we consider two cases: (a) if $\phi(\bot,\mseq{I})=\top$, then $\phi(\psi^\ell(\mseq{I}),\mseq{I})=\top$; (b) if $\phi(\bot,\mseq{I})=\bot$, then we have $\forall \ell\, \phi(\ell,\mseq{I})=\bot$. Therefore, in all cases, we have $\exists \ell\, \phi(\ell,\mseq{I})\iff\phi(\psi^\ell(\mseq{I}),\mseq{I})$. The above method of obtaining a Skolem function for a single system output is also called self-substitution~\cite{cav18,jiang,rsynth}.  

The second part of the theorem now follows from the observation that \Eup{G}{E}.
\end{proof}

\subsection{Generating Skolem functions from {\eqnf} circuits}
Theorem~\ref{main} suggests an efficient algorithm for generating a Skolem function vector from a specification given as a \eqnf\ circuit.  Algorithm~\ref{alg:skgen} presents the pseudo-code of algorithm $SkGen$.  The purpose of sub-routines used in $SkGen$ is explained in the comments. 

We illustrate the running of $SkGen$ by considering its execution on the circuit $G$ shown in Fig.~\ref{fig:eqNNF}. Here, $\mseq{X} = (\xone, \xtwo )$ and $\mseq{I}=(i)$.  As discussed earlier, we use $L_0$ through $L_{15}$ to denote the leaves of the circuit $G$ in left-to-right order. We have also seen earlier that $G$ is in \eqnf\ for the sequence of subsets of leaves $(S_1, S_2, S_3, S_4)$, where $S_1 = \{L_3\}$, $S_2 =\{L_7\}$, $S_3 =\{L_5\}$ and $S_4 =\{L_1\}$. 

\begin{algorithm}[t]
  \caption{\Genskolem($G, S, r$)}
    \label{alg:skgen}
        \KwIn{$G$: Relational spec in {\eqnf};\\  $\mathcal{S} = (S_1,S_2...S_k)$:  Sequence of $\land$-unrealizable subsets of $\lits(\mseq{X})$-labeled leaves of $C$;
        $r$: Recursion level}
    \KwOut{$\bpsi_G(\II)$: Skolem function vector for C}
    \uIf{ r = k+1}{
         $\bpsi_G(\mseq{I})$ := GetAnyFuncVec($|\mseq{X}|, \mseq{I}$)\;
         \tcp{Returns an $|\mseq{X}|-$dim vector of (arbitrary) functions of $\mseq{I}$}
        }
    \uElse{
    $\ell$ := Literal label of leaves in $S_r$\;
    $p_r$ := newOutputVar()
    \tcp*{$p_r$ is auxiliary output variable added at recursion level $r$}
    $E$ := GetCkt($G, S_r, \ell, p_r$)
    \tcp*{Replace all $\ell$-labeled leaves of $G$ other than those in $S_r$ by $\ell \wedge p_r$, and replace all $\neg \ell$-labeled leaves by $\neg \ell \wedge \neg p_r$}
    $S$ := GetNewSeq($S, r, \ell, p_r$)
    \tcp*{Replace $\ell$ by $p_r$ and $\neg \ell$ by $\neg p_r$ in all elements (leaves) of $S_j$ for $j > r$}
    $H$ := CPropSimp($E\mid_{\ell=\T,\neg\ell=\T}$)
    \tcp*{CPropSimp propagates constants and eliminates gates with constant outputs} 
        $\bpsi_H(\mseq{I}) = \Genskolem(H, S, r+1)$\; 
    
    $\psi_E^{\ell}(\mseq{I})$ := $\varphi_{E|_{\ell=\T,\neg \ell = \F}}(\bpsi_H(\mseq{I}),\mseq{I})$ 
    \tcp*{$\psi_E^{\ell}$ gives Skolem function for $\ell$ in $\varphi_E$} 
         {$\bpsi_G(\mseq{I}) = \big(\bpsi_H(\mseq{I}) \setminus (\psi_H^{p_r}), \psi_E^{\ell}(\mseq{I})\big)$}
    \tcp*{$\psi_H^{p_r}$ is Skolem function for $p_r$ in $\psi_H$}
    }
    \Return $\bpsi_G(\mseq{I})$\;
\end{algorithm}
As algorithm $SkGen$ proceeds, labels of different leaves of $G$ need to be updated.  For notational convenience, we use $G^{(r)}$, $H^{(r)}$ and $E^{(r)}$ to refer to the circuits $G$, $H$ and $E$ in the $r^{th}$ level of recursion of $SkGen$.   Table~\ref{tab:skgen-run} shows how $G^{(r)}, H^{(r)}$ and $E^{(r)}$ are obtained by replacing the labels of suitable leaves of $G$.  Each entry in this table lists which leaf labels of $G$ must be updated, where $L_{\{i,j,k\}}: f$ denotes updation of the label of each leaf in $\{L_i,L_j,L_k\}$ by $f$.  All leaves whose label updates are not specified are assumed to have the same labels as in $G$.
\begin{table}
\scriptsize
\begin{tabular}[t]{|p{0.5cm}|p{0.1\columnwidth}|p{0.3\columnwidth}|p{0.3\columnwidth}|}
\hline
 & $\mathbf{G^{(r)}}$ & $\mathbf{E^{(r)}}$ & $\mathbf{H^{(r)}}$ \\
 \hline
 $r$=$1$: & None  & $L_{\{1,14\}}$:$(\neg \xone \wedge \neg p_1)$, 
                        $L_{\{10\}}$:$(\xone \wedge p_1)$ 
                             & $L_{\{1,14\}}$:$\neg p_1$, 
                               $L_{\{3\}}$:$\T$, $L_{\{10\}}$:$p_1$ \\
  \hline        
 $r$=$2$: & Same as in $H^{(1)}$ 
              & $L_{\{1,14\}}$:$\neg p_1$, 
                 $L_{\{3\}}$:$\T$, $L_{\{10\}}$:$p_1$,
                 $L_{\{5,9,11,15\}}$:$\neg \xtwo$$\wedge$$\neg p_2$,
                 $L_{\{13\}}$:$\xtwo \wedge p_2$ 
                     & $L_{\{1,14\}}$:$\neg p_1$, 
                       $L_{\{3,7\}}$:$\T$, $L_{\{10\}}$:$p_1$,
                       $L_{\{5,9,11,15\}}$:$\neg p_2$,    
                       $L_{\{13\}}$:$p_2$ \\
 \hline
 $r$=$3$: & Same as in $H^{(2)}$ 
              & $L_{\{1,14\}}$:$\neg p_1$, 
                $L_{\{3,7\}}$:$\T$, $L_{\{10\}}$:$p_1$,
                $L_{\{9,11,15\}}$:$\neg p_2$$\wedge$$p_3$, 
                $L_{\{5\}}$:$\neg p_2$,
                $L_{\{13\}}$:$p_2$$\wedge$$\neg p_3$
                   & $L_{\{1,14\}}$:$\neg p_1$, 
                     $L_{\{3,5,7\}}$:$\T$, $L_{\{10\}}$:$p_1$,
                     $L_{\{9,11,15\}}$:$p_3$, 
                     $L_{\{13\}}$:$\neg p_3$ \\
\hline
 $r$=$4$: & Same as in $H^{(3)}$ 
             & $L_{\{1\}}$:$\neg p_1$, 
               $L_{\{14\}}$:$\neg p_1$$\wedge$$p_4$, 
               $L_{\{3,5,7\}}$:$\T$, $L_{\{10\}}$:$p_1$$\wedge$$\neg p_4$,
               $L_{\{9,11,15\}}$:$p_3$, 
               $L_{\{13\}}$:$\neg p_3$ 
                 & $L_{\{14\}}$:$p_4$, 
                   $L_{\{1,3,5,7\}}$:$\T$, $L_{\{10\}}$:$\neg p_4$,
                   $L_{\{9,11,15\}}$:$p_3$, 
                   $L_{\{13\}}$:$\neg p_3$ \\
\hline
\end{tabular}
\caption{Run of Algorithm~\ref{alg:skgen} on Fig.~\ref{fig:eqNNF}}
\label{tab:skgen-run}
\normalsize
\end{table}
It can be verified that $G$ with leaf labels updated as in the table entry corresponding to $H^{(4)}$ simplifies to $\T$ by constant propagation.  Hence $H^{(4)}$ is semantically independent of $\{x_1, x_2\}$. This is not a coincidence, but is guaranteed by the definition of \eqnf. Hence, at recursion level $5$ of $SkGen$, any vector of functions $\big(f_3(i), f_4(i)\big)$ can be returned in line $2$ of Algorithm~\ref{alg:skgen} as a Skolem function vector for $(p_3, p_4)$ in $G^{(5)} = H^{(4)}$. 

As the recursive calls return, we obtain $E^{(4)}(p_3 = f_3(i), p_4= f_4(i), p_1=\T, i)$ as Skolem function for $p_1$ in $E^{(4)}$. Call this function $f_1(i)$.  Next, we get $f_2(i) = E^{(3)}(p_1=f_1(i), p_3 = f_3(i), p_2=\T, i)$ as Skolem function for $p_2$ in $E^{(3)}$. Continuing further, we obtain $f_{\xtwo(i)} = E^{(2)}(p_1 = f_1(i), p_2=f_2(i), \xtwo=\T, i)$ as Skolem function for $b$ in $E^{(2)}$, and $f_{\xone(i)} = E^{(1)}(\xtwo=f_{\xtwo(i)}, p_1=f_1(i), \xone = \T, i)$ as Skolem function for $\xone$ in $E^{(1)}$. The final
return gives $\big(f_{\xone(i)}, f_{\xtwo(i)}\big)$ as a Skolem function vector for $(\xone, \xtwo)$ in $G$.  Note that different choices of $f_3(i), f_4(i)$ yield different Skolem function vectors of $G$, all of which are correct.

\begin{theorem}\label{thm:skolem}
Suppose Algorithm $SkGen$ is invoked with a circuit $G$ as input, that is in \eqnf\ w.r.t. $\mset{\XX}$ and a sequence $S$ of subsets of leaves. Assuming the vector returned in line $2$ of the algorithm can be constructed in time $\bigO{|S|^2\cdot|G|}$, the algorithm returns a Skolem function vector of size in $\bigO{|S|^2\cdot|G|}$ in time $\bigO{|S|^2\cdot|G|}$.
\end{theorem}
\begin{proof}
We show that Algorithm $SkGen$ generates a  Skolem function vector by an inductive application of Theorem~\ref{main}.
Specifically, for each level $i$ of recursion, if the Skolem function $\bpsi_H$ returned in line $9$ by the $i+1^{st}$ recursive
call of $SkGen$ is correct for $H$,
Theorem~\ref{main} ensures that the Skolem function $\bpsi_G$ computed in lines $10$ and $11$ of the $i^{th}$ recursive call is correct for $G$.

To see why the terminating case of this recursion yields correct Skolem functions, note that when the
recursion level is $k+1$ (lines $1$-$2$ of Algorithm~\ref{alg:skgen}), by Definition~\ref{def:saunf}, the function represented by $G$ is semantically independent of $\mseq{X}$.  Hence any Skolem function vector for $\mseq{X}$ suffices in line $2$ of Algorithm~\ref{alg:skgen}.  

Algorithm $SkGen$ has exactly $k+1$ recursive calls, and in each of the first $k$ calls, the steps in lines $4, 5, 6, 7$ and $8$ take time linear in $|G|$ and generate circuits $E$ and $H$ that are of size in $\mathcal{O}(|G|)$.  Indeed, the circuit $H$ in each recursion level $\le k$ is simply $G$ with the literal labels of some of its leaves replaced by other literals or by Boolean constants (possibly followed by simplification via constant propagation).  Therefore, $|H| \leq |G|$ in each recursion level $\le k$.  The circuit $E$ is similarly obtained by replacing some leaves of $G$
with Boolean constants, other literals or conjunctions of two literals.  Therefore, $|E| \leq 2\times|G|$ in
each recursion level $\le k$.  In the $k+1^{th}$ recursive call, line $2$ is executed, and as discussed above, we restrict it to take time in $\bigO{|S|^2\cdot|G|}$.  This also ensures that the size of the Skolem function vector returned in line $2$ is in $\bigO{|S|^2\cdot|G|}$. 

Once the recursive calls start returning, lines $10$ and $11$ of Algorithm $SkGen$ are executed.  Note that in line $10$, the Skolem function for $\ell$ is obtained by feeding into the inputs of circuit $E$ (as obtained in the current level of recursion) the outputs of Skolem functions computed in later (or higher) levels of the recursion.  We have already seen above that $|E|$ is at most $2\times|G|$ in each recursion level $\le k$.  A Skolem function computed at recursion level $j$ can potentially feed into $E$ at all recursion levels in $\{1,\ldots j-1\}$.  Therefore, a maximum of $\sum_{j=2}^{|S|+1} (j-1) \in \bigO{|S|^2}$ connections may need to be created between the output of a Skolem function generated at some recursion level and the input of $E$ at a lower level of recursion.  Therefore, constructing the entire Skolem function vector at recursion level $1$ requires time (and, hence space) in $\bigO{|S|^2\cdot|G|}$.
\end{proof}
In the above analysis, we assumed that the vector of functions used in line $2$ of Algorithm~\ref{alg:skgen} has size in $\bigO{|S|^2\cdot|G|}$.
Since an arbitrary vector of functions of $\mseq{I}$ suffices in line $2$, we can choose a $|\mseq{X}|$-dimensional constant function vector, say $(\F, \ldots \F)$  as the output of {\tt GetAnyFuncVec}.  Hence, the above assumption can \emph{always} be satisfied. Note that Theorem~\ref{thm:skolem} guarantees that $\eqnf$ enjoys property P1 of Section~\ref{sec:prob-stmt}.

\subsection{Generating {\eqnf} circuits from Skolem functions}
Next, we show that if we already know one (out of possibly many) Skolem function vector of a relational specification $G$ given as a circuit, we can easily derive a semantically equivalent circuit in \eqnf.
\begin{theorem}\label{lemma:structure}
Let $\bpsi(\mseq{I})$ be a Skolem function vector for $\mseq{X}$ in $\varphi_G(\XX,\II)$. Define $G'$ to be $G$ with all labels $x_i$ (resp. $\neg x_i$) in $\lits(\mseq{X})$ replaced by $\psi_i(\mseq{I})$ (resp. $\neg \psi_i(\mseq{I})$), i.e. $\varphi_{G'}(\mseq{I}) = \varphi_G\big(\bpsi(\mseq{I}), \mseq{I})$.  
Define $F$ to be the circuit representing the formula $\bigwedge_{i=1}^m \big((x_i\land \psi_i(\II)) \lor (\neg{x_i}\land \neg{\psi_i}(\II))\big)$, and  $H$ to be the circuit representing $\big(\varphi_{F}\lor \varphi_G\big) \land \varphi_{G'}$. Then $H$ is in \eqnf\ w.r.t. $\mset{\XX}$ and a sequence of literal-consistent subsets of leaves, and $\varphi_G(\mseq{X}, \mseq{I}) \iff \varphi_H(\mseq{X}, \mseq{I})$.
\end{theorem}
\begin{proof}
 We first show that $\varphi_G(\mseq{X}, \mseq{I}) \iff \varphi_H(\mseq{X}, \mseq{I})$.  This involves showing two implications.
 \begin{itemize}
     \item $\varphi_H(\XX,\II) \Rightarrow \varphi_G(\XX,\II):$ We know from the definitions that $\varphi_{G'}(\II) \Leftrightarrow \varphi_G(\bpsi(\II),\II)$ $\Rightarrow \forall \XX\, (\varphi_{F}(\XX,\II)\Rightarrow \varphi_G(\XX,\II))$.  The last implication follows from the observation that $\varphi_{F}(\XX,\II)$ simply asserts that $\bigwedge_{i=1}^m \big(x_i \Leftrightarrow \psi_i(\II)\big)$ holds, and hence $\varphi_G(\bpsi(\II), \II)  \Rightarrow  \forall \XX\, \big(\varphi_F(\XX,\II) \Rightarrow \varphi_G(\XX, \II)\big)$.  We also know from the definition of $H$ that  $\varphi_H(\XX,\II)\Rightarrow \varphi_{G'}(\II)$ and $\varphi_H(\XX,\II) \Rightarrow \big(\varphi_F(\XX,\II)\lor \varphi_G(\XX,\II)\big)$.  However, since $\varphi_{G'}(\II) \Rightarrow \forall \XX\, \big(\varphi_F(\XX,\II) \Rightarrow \varphi_G(\XX, \II)\big)$, it follows that $\varphi_H(\XX, \II) \Rightarrow \big(\varphi_{G'}(\II)\land \varphi_G(\XX,\II)\big)$.  Hence $\varphi_H(\XX, \II) \Rightarrow \varphi_G(\XX,\II)$. 
     \item $\varphi_G(\XX, \II) \Rightarrow \varphi_H(\XX,\II)$:  We know that $\varphi_G(\XX,\II)\Rightarrow$ $\exists \XX\, \varphi_G(\XX, \II)$  $\Leftrightarrow \varphi_G(\bpsi(\II), \II)$ $\Leftrightarrow \varphi_{G'}(\II)$ by definition. It follows that $\varphi_G(\XX,\II)\Rightarrow \big(\varphi_G(\XX,\II) \land \varphi_{G'}(\II)\big)$.  However, from the definition of $H$, we know that $\big(\varphi_G(\XX,\II) \land \varphi_{G'}(\II)\big) \Rightarrow \varphi_H(\XX,\II)$.  Hence, $\varphi_G(\XX,\II) \Rightarrow \varphi_H(\XX,\II)$. 
 \end{itemize}
 
To show that $H$ is in \eqnf\ w.r.t. $\mset{\XX}$ and a suitably defined sequence $S$ of subsets of leaves, we first observe that the circuit $F$ that naturally represents $\bigwedge_{i=1}^n \big((x_i\land \psi_i(\II)) \lor (\neg{x_i}\land \neg{\psi_i}(\II))\big)$ has a $\wedge \vee \wedge$-structure with leaves labeled by $\lits(\XX)$ and $\lits(\II)$.  It is easy to see from the structure of this circuit that for $i \in \{1, \ldots m\}$, literals $x_i$ (and $\neg x_i$) are $\land$-unrealizable in $F$.  Let $S = (S_1, S_2, \ldots S_{2m})$ be a sequence of subsets of leaves of $F$, where  $S_{2i-1}$ is the set of all $x_i$-labeled leaves of $F$, and $S_{2i}$ is the set of all $\neg x_i$-labeled leaves labeled of $F$.  Since the literals $x_i$ and $\neg x_i$ are  $\land$-unrealizable in $F$ for $i \in \{1, \ldots m\}$, we also have that  $S_j$ is $\land$-unrealizable in $F$ for every $j \in \{1, \ldots 2m\}$.  Finally, $F|_{S_1:\T, \ldots S_{2m}:\T}$ has no literal in $\lits(\XX)$ labeling any leaf. Hence, $F$ satisfies all the conditions of Definition~\ref{def:saunf}, and is in {\eqnf} w.r.t. $\mset{\XX}$ and $S$.

We now claim that the circuit, say $R$, representing $\big(\varphi_{F} \vee \varphi_G\big)\wedge \varphi_{G'}$ is in {\eqnf} w.r.t. $\mset{\XX}$ and the sequence $S$ of subsets of leaves of $F$ described above.  To see why this is so, observe that all subsets $S_j \in S$ are mutually disjoint and contain leaves labeled by $\lits(\XX)$.  Hence the first 2 conditions of Definition~\ref{def:saunf} are satisfied.  To see why conditions 3 and 4 of Definition~\ref{def:saunf} are satisfied, recall from Definition~\ref{def:set-realizability} that when checking $\land$-unrealizability of any $S_i \in S$ in $R$, all leaves of the sub-circuit $G$ (of the circuit $R$) that are labeled by the same literal as leaves in $S_i$, must be re-labeled $\F$.  This, coupled with the fact that $S_i$ is $\land$-unrealizable in $F$, ensures that $S_i$ is $\land$-unrealizable in $R$ as well.  Finally, as we will see in Section~\ref{sec:operations} (see Lemma~\ref{exist-quantification}), 
$\varphi_{F\mid_{x_1 = \T, \neg x_1 = \T, \ldots x_n = \T, \neg x_n = \T}}$ $\Leftrightarrow$ $\exists \XX\, \varphi_F(\XX, \II) $ $\Leftrightarrow \T$. The last equivalence follows from the definition of $F$; specifically $\varphi_F(\bpsi(\II), \II) = \T$ for all assignments of $\II$. Therefore $\varphi_{R|_{S_1:\T, \ldots S_{2m}:\T|}} \Leftrightarrow (\T \vee \varphi_G) \wedge \varphi_{G'}$
$\Leftrightarrow \varphi_{G'}$.
Since $\varphi_{G'}$ is semantically independent of $\XX$ by definition, condition 5 of Definition~\ref{def:saunf} is satisfied for $R$. Hence, $R$ is in \eqnf\ w.r.t. $\mset{\XX}$ and the sequence $S = (S_1, S_2,\ldots S_{2m})$.
\end{proof}
The proof of Theorem~\ref{lemma:structure} gives the following corollary.
\begin{corollary}\label{cor:eqnf-constr}
Given $G$, and a Skolem function vector $\bpsi(\II)$ for $\XX$ in $\varphi_G(\XX,\II)$, a {\eqnf} circuit $H$ semantically equivalent to $G$ can be constructed in $\bigO{|G|+|\bpsi|}$ time. Furthermore, $|H| \in \bigO{|G|+|\bpsi|}$.
\end{corollary}

Finally, Corollary~\ref{cor:eqnf-constr} and Theorem~\ref{thm:skolem} immediately yield the following theorem.
\begin{theorem}\label{thm:characterize}
For every class $\mathcal{C}$ of circuits representing relational specifications,
\begin{enumerate}
    \item Boolean Skolem function synthesis can be solved in polynomial-time for $\mathcal{C}$ iff every circuit in $\mathcal{C}$ can be compiled to a semantically equivalent {\eqnf} circuit in polynomial-time.
    \item A Skolem function vector of polynomial size exists for every specification in $\mathcal{C}$ iff every circuit in $\mathcal{C}$ can be compiled to a polynomial-sized semantically equivalent {\eqnf} circuit.
    \end{enumerate}                                                 
\end{theorem}
Thus, {\eqnf} satisfies properties P2a and P2b of Section~\ref{sec:prob-stmt}.
In other words, {\eqnf} truly \emph{characterizes} efficient Boolean functional synthesis.  Note that Theorem~\ref{thm:characterize} is significantly stronger than sufficient conditions for efficient synthesis given in~\cite{cav18,fmcad19}. 


\section{Operations on \eqnf}\label{sec:operations}
In this section, we discuss the application of basic operations like conjunction, disjunction and existential quantification of variables on formulas represented by \eqnf\ circuits, and also examine the complexity of checking if a given circuit is in \eqnf. Throughout the section, we assume that all specifications (circuits) are over system inputs $\mset{\mseq{I}}$ and system outputs $\mset{\mseq{X}}$ unless otherwise specified.  To reduce notational clutter, given circuits $G$ and $H$, we abuse notation and use $G \vee H$ (resp. $G \wedge H$) to denote the circuit consisting of an $\lor$- (resp. $\land$-)labeled root node with the roots of $G$ and $H$ as its children.
 \begin{lemma}\label{disjunction}
Suppose $G$ is in \eqnf\ w.r.t $\mset{\mseq{X}}$ and a sequence $S^G$, and $H$ is in \eqnf\ w.r.t. $\mset{\mseq{X}}$ and $S^H$. Then the circuit $G\lor H$ is in \eqnf\ w.r.t $\mset{\mseq{X}}$ and $(S^G, S^H)$.
\end{lemma}
\begin{proof}
The proof follows from the claim that $(S^G, S^H)$ is a sequence of subsets of leaves of $G\lor H$ 
that satisfies the conditions of Definition~\ref{def:saunf}.  To see why this is so, note that by Definition~\ref{def:set-realizability}, when considering a subset, say $S^G_i$, of $\ell$-labeled leaves in $S^G$, all $\ell$-labeled leaves of $H$ must be re-labeled $\bot$.  Hence, $H$ can only contribute $\neg{\ell}$, and can, at worst, combine with $\ell$ contributed by $G$ at the $\lor$-labeled root of $G\lor H$. Since the set of $\ell$-labeled leaves in $S^G_i$ are already $\land$-unrealizable in $G$, we find that $S^G_i$ is $\land$-unrealizable in $G \lor H$ as well. By repeating this argument, we find that conditions 1, 2, 3 and 4 of Definition~\ref{def:saunf} are satisfied by the circuit $G \lor H$.  To see why condition 5 is also satisfied, observe that since $G$ (resp. $H$) is in {\eqnf} w.r.t $\mset{\XX}$ and $S^G$ (resp. $S^H$), when all subsets of leaves in $(S^G, S^H)$ are re-labeled to $\top$, the circuit $G \lor H$ represents the disjunction of two formulas, each of which is semantically independent of $\mseq{X}$.  Hence, $G \lor H$ is in {\eqnf} w.r.t $\mset{\mseq{X}}$ and $(S^G,S^H)$.
\end{proof}
Significantly, Lemma~\ref{disjunction} does not require any assumptions on the relation between ordering of subsets in $S^G$ and $S^H$.  Other popular normal forms do not enjoy this property.
For example,  disjoining two {\ROBDD}s constructed with different ordering of variables does not always yield an {\ROBDD} in polynomial-time.
Similarly, combining two \snf\ circuits with an $\lor$ gate may not result in a \snf\ circuit unless the ordering of output variables in both circuits are the same. This shows that disjunction is more efficiently computable in \eqnf\ than in {\ROBDD}s or even in \snf.

Next, we observe that existential quantification of \emph{all system outputs} is easy for \eqnf\ representations.
\begin{lemma}\label{exist-quantification}
Suppose $G$ is in \eqnf\ w.r.t $\mset{\mseq{X}}$ and a sequence $S$.  Let $L$ be the set of all leaves of $G$  that are labeled by a literal over $\mseq{X}$. Then $\exists \mseq{X}\, G \Leftrightarrow G\mid_{L:\top}$.
\end{lemma}
\begin{proof}
Follows from Theorem~\ref{main}(1) and Definition~\ref{def:saunf}.
\end{proof}

Next, we move to the more difficult case of conjunction.
\begin{lemma}
Suppose $G$ is in \eqnf\ w.r.t $\mset{\mseq{X}}$ and a sequence $S^G$, and $H$ is in \eqnf\ w.r.t. $\mset{\mseq{X}}$ and sequence $S^H$. If there is no literal $\ell$ over $\mseq{X}$ such that $G$ has an $\ell$-labeled leaf and $H$ has a $\neg\ell$-labeled leaf, the circuit $G \land H$ is in \eqnf\ w.r.t. $\mset{\mseq{X}}$ and $(S^G,S^H)$. Otherwise, a \eqnf\ circuit semantically equivalent to $G\land H$ cannot be constructed in time polynomial in $|G|,|H|$ unless $\myP = \myNP$.  Further, such a circuit cannot have size polynomial in $|G|,|H|$ unless $\Pi_2^P = \Sigma_2^P$ (i.e. unless the polynomial hierarchy collapses to the second level). 
\end{lemma}
\begin{proof}
If there is no literal $\ell$ over $\mseq{X}$ such that $G$ has an $\ell$-labeled leaf and $H$ has a $\neg\ell$-labeled leaf, it is easy to see that a leaf of $G$ and a leaf of $H$ cannot participate together to make any literal in $\lits(\XX)$ $\land$-realizable in the circuit $G \wedge H$.  Since $G$ is in \eqnf\ w.r.t. $\mset{\mseq{X}}$ and $S^G$, and $H$ is in \eqnf\ w.r.t. $\mset{\mseq{X}}$ and $S^H$, it then follows that $G \land H$ is also in \eqnf\ w.r.t. $\mset{\mseq{X}}$ and the sequence $(S^G, S^H)$ (or alternatively, $(S^H, S^G)$).

The proof for the remainder of the lemma is more intricate. For this part of the discussion, we consider circuits $G$ in which the labels of all leaves are considered to be from $\lits(\XX)$. In other words $\II$ is assumed to be empty.  Let $G^-$ denote the circuit $G|_{\neg x_1=x'_1,\ldots \neg x_m=x'_m}$ where $\mseq{X'}=(x'_1,\ldots,x'_m)$ is a sequence of fresh variables. This is sometimes called the \emph{positive form} of the circuit. 

\begin{claim}\label{fc:syn}
For every circuit $G$, the circuit $G^{-}$
is in ${\eqnf}$ w.r.t. 
$\mset{\mseq{X}} \cup \mset{\mseq{X}'}$ for any sequence of literal-consistent subsets of leaves.
\end{claim}
\begin{proof}
There is no label of a leaf of $G^{-}$ whose negation is also the label of some other leaf of $G^{-}$ (since there are no negated output literals at all in the labels of leaves).  This eliminates the possibility of a subset of literal-consistent leaves being $\land$-realizable in $G^{-}$.
\end{proof}
Further, given $G$, let  $G^{+}$ be the circuit representing the formula
$\land_{i=1}^{n}((x_i' \land \neg{x_i})\lor(\neg{x_i'} \land x_i))$ where $x_i'$ are the variables introduced in $G^{-}$. Note that $\varphi_{G^{+}}$ $\Leftrightarrow$ $\land_{i=1}^{n}(x_i'\Leftrightarrow \neg{x_i})$
\begin{claim}\label{gc:syn}
For every circuit $G$, the circuit $G^{+}$ is in \eqnf\ w.r.t. $\mset{\mseq{X}} \cup \mset{\mseq{X'}}$ and the sequence of literal-consistent leaves $S = (S_1, \ldots S_{2n}, S_1', \ldots S_{2n}')$, where $S_{2i-1}$ (resp. $S_{2i-1}')$ is the set of all leaves labeled $x_i$ (resp. $x_i'$), and $S_{2i}$ (resp. $S_{2i}$') is the set of all leaves labeled $\neg x_i$ (resp. $\neg x_i'$) in $G^{+}$.
\end{claim}
\begin{proof}
It is easy to see that for the sequence $(S_1, S_2, \ldots S_{2n})$, the first 4 conditions of Definition~\ref{def:saunf} are satisfied. Moreover $\varphi_{G^{+}\mid_{S_1:\T, S_2:\T, \ldots S_{2n}:\T}} \Leftrightarrow \T$.  Therefore, condition 5 of Definition~\ref{def:saunf} is also satisfied, and $G^{+}$ is in \eqnf\ w.r.t. $\mset{\mseq{X}}\cup \mset{\mseq{X'}}$ and  $(S_1, \ldots S_{2n})$, and hence also w.r.t. the sequence $S = (S_1, \ldots S_{2n}, S_1', \ldots S_{2n}')$.
\end{proof}
It is easy to see  that $\varphi_{G^-}\land \varphi_{G^{+}}$ is equisatisfiable to $\varphi_G$. Now, consider an arbitrary instance of the Boolean satisfiability (SAT) problem, i.e., given a Boolean circuit $G$ over $\XX=(x_1,\ldots x_m)$, we must determine if $\varphi_G$ is satisfiable. We interpret $\varphi_G$ as a relational specification over system outputs $\mseq{X}$, with the system inputs $\mseq{I}$ being absent. 

  By definition, $|\Fc|$ and $|\Gc|$ are in $\bigO{|G|}$.  Using Claims~\ref{fc:syn} and~\ref{gc:syn}, each of these circuits is also in $\eqnf$ w.r.t. $\mset{\XX} \cup \mset{\XX'}$ and an appropriate sequence of subsets of leaves. Suppose there exists a polynomial-time algorithm $A$ that takes two \eqnf\  circuits as inputs and produces a {\eqnf}
    circuit semantically equivalent to the conjunction of the formulas represented by the
    two circuits.   We use algorithm $A$ to obtain a circuit $\widehat{G}$ that is semantically equivalent to $\llbracket \varphi_{\Fc} \wedge \varphi_{\Gc}\rrbracket$.  Clearly, $|\widehat{G}|$ must have size polynomial in $|\Fc|$ and $|\Gc|$, and therefore polynomial in $|G|$.  
    Since every \eqnf\ circuit yields Skolem functions for all outputs in time polynomial in the size of the circuit (see Theorem~\ref{thm:skolem}), we can compute a Skolem function for every output of $\widehat{G}$ in time polynomial in $|G|$.  Since there are no inputs, each of these Skolem functions must simplify to a Boolean constant.  From the definition of Skolem functions, we also know that $\widehat{G}$ is satisfiable iff the Skolem functions obtained above (Boolean constants for variables in $\mset{\mseq{X}} \cup \mset{\mseq{X'}}$) cause $\llbracket \widehat{G} \rrbracket$ to evaluate to $\T$.  In other words, we can determine if $\varphi_{\widehat{G}}$, and hence $\varphi_{\Fc} \wedge \varphi_{\Gc}$, is satisfiable in time polynomial in $|G|$.   Since $\varphi_{\Fc}\land \varphi_{\Gc}$ is equisatisfiable to $\varphi_G$, this effectively solves the Boolean satisfiability problem in polynomial time. Therefore, algorithm $A$ cannot run in polynomial time unless $\myP=\myNP$.
    
    Suppose for every two $\eqnf$ circuits, there exists a polynomial sized {\eqnf} circuit that is semantically equivalent to the conjunction of the formulas represented by the two circuits. Let $\widetilde{G}$ be the circuit obtained in this manner for $\llbracket \varphi_{\Fc} \wedge \varphi_{\Gc}\rrbracket$.
    By Theorem~\ref{thm:skolem}, Skolem functions synthesized from $\widetilde{G}$ must have size polynomial in $|G|$.  By the same argument as above, it now follows that Boolean satisfiability must be in {\myPSPoly}.  This implies that $\myNP \subseteq \myPSPoly$.  By Karp-Lipton Theorem, however, we know that this implies that the polynomial hierarchy collapses to the second level, i.e. $\Pi_2^P = \Sigma_2^P$.
\end{proof}

Finally, we ask how difficult it is to check if a given circuit $G$ is in {\eqnf}.
\begin{theorem}\label{thm:membership}
\begin{enumerate}
    \item Given $G$ and a sequence $S$ of literal-consistent disjoint subsets of leaves, checking if $G$ is in \eqnf\ w.r.t. $\mset{\mseq{X}}$ and $S$ is {\mycoNP}-complete.
    \item Given $G$, checking if $G$ is in \eqnf\ w.r.t. $\mset{\mseq{X}}$ and some (unspecified) sequence of subsets of leaves is {\mycoNP}-hard and in $\Sigma_2^P$.
\end{enumerate}
\end{theorem}
\begin{proof}
 First, we show that identifying whether a given circuit is in \eqnf\ for a given sequence of subsets of leaves is in coNP. This is equivalent to asking whether the complement problem, i.e. if the given circuit is \textit{not} in \eqnf\ for the given sequence of subsets of leaves, is in \myNP. We will define a non-deterministic polynomial-time Turing machine $M$ that solves this complement problem. The machine $M$ first checks whether the input circuit (say $G$) is in NNF.  This can be done by checking if each internal node of $G$ is labeled either $\land$ or $\lor$, and if all negations (if any) are on the labels of leaves.  Clearly, this check can be done in time polynomial in the size of $G$.  If the circuit is found to be not in NNF, the machine $M$ accepts, since $G$ cannot be in \eqnf\ in this case.  Otherwise (i.e. if $G$ is in NNF), the machine $M$ non-deterministically chooses a subset $S_i$ of literal-consistent leaves in the given sequence $S$ and executes the following operations.  
    
    Suppose the literal labeling leaves in the subset $S_i$ is $\ell_i$. The machine $M$ does the following: (a) it constructs $G^{\dagger} = G\mid_{S_1:\top\ldots S_{i-1}:\top}$, (b) sets all leaves of $G^{\dagger}$ that are not in $S_i$ but are labeled $\ell_i$ to $\bot$, (c) replaces all remaining labels $\ell$ (resp. $\neg \ell$) on leaves by $w$ (resp. $w'$), (d) guesses an assignment $\sigma$ to all variables other than $w$ and $w'$ labeling leaves in the resulting circuit,  and (e) checks if the resulting circuit represents the Boolean function $w \land w'$ for the assignment $\sigma$.   Note that after step (d), the resulting circuit represents a function of only $w$ and $w'$.  Hence the check in step (e) can be performed by setting $(w, w')$ to each of $(\top,\top), (\top,\bot), (\bot,\top)$ and $(\bot,\bot)$ and checking if the resulting circuit evaluates to $\top$, $\bot$, $\bot$ and $\bot$ respectively. Clearly, all the steps above can be done in time polynomial in $|G|$.  If after step (e), the resulting circuit is found to represent $w \land w'$, then machine $M$ accepts.  In this case, $G$ is not in \eqnf\ w.r.t. $\mset{\mseq{X}}$ and the given sequence $S$ of subsets of leaves.  Conversely, if the circuit is not in \eqnf\ w.r.t. $\mset{\mseq{X}}$ and the given sequence $S$, then there is a subset $S_i$ of leaves in $S$ that is $\land$-realizable in $G$ for the assignment described above. Hence the problem of identifying whether a circuit is \emph{not} in \eqnf\ for a given sequence of subsets of leaves is in \myNP.  Thus,  the problem of identifying whether $G$ is in \eqnf\ for a given sequence of subsets of leaves is in \mycoNP.\\
\indent Next, we show that the problem is co-NP hard. We reduce the problem of checking if a propositional formula represented by a \CNF\ circuit $G$ is unsatisfiable to identifying whether an appropriately constructed circuit is in \eqnf\ for a specific sequence of subsets of literal-consistent leaves. For this, we consider the specification  $G\land x\land \neg{x}$, where $x$ is the sole output of the specification, and the inputs are the variables labeling leaves of $G$. Since there is only one output variable, there are only two (equivalent) orderings of subsets of leaves labeled by output literals. It is easy to see that $x$ (equivalently, $\neg x$) is $\land$-realizable if and only if $G$ is satisfiable. Hence, identifying whether a problem is in \eqnf\ for a given sequence of subsets of leaves is coNP-hard.

To prove the second part of the theorem, we show that checking whether a given $G$ is in \eqnf\ w.r.t. some (unspecified) sequence of subsets of output literal-consistent leaves can be solved by a non-deterministic polynomial-time Turing machine $M$ with access to an {\myNP} oracle, i.e. the problem is in ${\myNP}^{\myNP}$.  Given a specification $G$ with system inputs $\mseq{I}$ and system outputs ${\mseq{X}}$ the machine $M$ does the following:
\begin{enumerate}
    \item[(i)] It guesses a sequence $S$ of disjoint subsets of literal-consistent leaves.
    \item[(ii)] It then reduces the problem of deciding whether $G$ is not in \eqnf\ w.r.t $\mset{\mseq{X}}$ and the sequence $S$ to checking the satisfiability of an appropriately constructed propositional formula $\varphi$.  This reduction is similar to what we discussed above in the proof of part (1).
    \item[(iii)] Finally, it feeds $\varphi$ to the \myNP oracle and accepts if and only if the \myNP oracle rejects.
\end{enumerate}
Therefore, $M$ accepts if and only if there is a sequence $S$ of subsets of output literal-consistent leaves for which $G$ is in \eqnf\ w.r.t. $\mset{\mseq{X}}$ and $S$. Hence, we have proved that our problem is contained in $\myNP^{\myNP}$, or equivalently in $\Sigma_2^P$. The proof that the problem is {\mycoNP}-hard uses arguments similar to the earlier {\mycoNP} hardness proof.
\end{proof}


\section{Conversion to \eqnf}\label{sec:conversion}
We now present an algorithm for compiling a circuit $G$ representing a {\CNF} formula over $\XX$ and $\II$ to a semantically equivalent circuit in {\eqnf}.  Algorithm {\tt GetSaunf} (see Algorithm~\ref{alg:getsaunf}) takes $G$ as input and produces a circuit $F$ and sequence $S$ of subset of leaves such that $\llbracket \varphi_G \rrbracket = \llbracket \varphi_F \rrbracket$ and $F$ is in {\eqnf} w.r.t. $\mset{\XX}$ and $S$.
Algorithm {\tt GetSAUNF} uses a routine named {\tt GetSubset} (shown in Algorithm~\ref{alg:getSubset}) to obtain an $\land$-unrealizable subset $U$ of leaves labeled by a chosen literal $\ell$.  This set is used to decompose the problem into two sub-problems: (i) a circuit $G'$ representing conjunction of all clauses that have $\ell$-labeled leaves in $U$ feeding them, and (ii) a circuit representing conjunction of all other clauses.  While $G'$ does not require any recursive application of {\tt GetSAUNF}, the second sub-problem is recursively solved using Shannon-style decomposition.  Finally, the sequence $S$ of subsets of leaves is obtained by suitably interleaving the subset $U$ and the sequences obtained from recursive applications of {\tt GetSAUNF}.  

\begin{algorithm}[t]
  \caption{\getSAUNF($G$)}
    \label{alg:getsaunf}
        \KwIn{$G$: Circuit representing a {\CNF} formula}
    \KwOut{$F$: \eqnf\ ckt semantically equivalent to $G$\\
    $S$: Sequence of literal-consistent subsets of leaves of $F$}
    
    \uIf{$\llbracket G \rrbracket$ is semantically independent of $\XX$}{\Return $\big(G, \emptyset\big)$\;}
    
    $\ell$ := ChooseLiteral($G$) \tcp*{Gives a literal over $\XX$ labeling a leaf in $G$}
    $U$ := \getSubset($G$, $\ell$)
    \tcp*{$U$ is an $\land$-unrealizable subset of $\ell$-labeled leaves in $G$}
    $G'$ := CktFromClausesWithLeaves($U$)
    \tcp*{$G'$ represents conjunction of clauses (in the {\CNF} formula represented by $G$) containing leaves in $U$}
    $D$ := $G\mid_{U:\T}$\;
    \uIf{$\llbracket D \rrbracket$ is semantically independent of $\ell$}{
    $(G_1, U_1)$ := \getSAUNF($D$)\;
    \Return $\big(G'\land G_1$, $\big(U, U_1\big)\big)$\;}
    
    $(G_1, U_1)$ := \getSAUNF($D_{|\ell=\T}$)\;
    $(G_2, U_2)$ := \getSAUNF($D_{|\ell=\F}$)\;
    $F$ := $G' \land ((\ell\land G_1)\lor (\neg{\ell}\land G_2))$\;
    $U_\ell$ := $\{ \ell$-leaf of $F$ in sub-circuit for $(\ell \land G_1)$ \}\;
    $U_{\neg\ell}$ := $\{\neg\ell$-leaf of $F$ in sub-circuit for $(\neg\ell \land G_2)$ \}\;
    \Return $\big(F, (U_\ell, U, U_{\neg\ell}, U_1, U_2)\big)$\;
\end{algorithm}

To understand how subroutine {\tt GetSAUNF} works, let $v_\ell$ denote the underlying variable of the literal $\ell$.  We use $D$ to represent the current view of the circuit (formula) from which we wish to extract the $\wedge$-unrealizable subset of leaves.  We also use $CurrS$ to denote a subset of clauses containing $\ell$ such that there exists an assignment $\sigma:\mset{\XX}\setminus\{v_\ell\} \cup \mset{\II} \rightarrow \{\F,\T\}$ for which all (and only) these clauses of the underlying {\CNF} formula evaluate to $\ell$, and $\ell$ is $\land-$realizable in $G$ under $\sigma$ (see Definition~\ref{def:l-unreal}).  Such an assignment $\sigma$ can be obtained by effectively finding a satisfying assignment of $\varphi_{D|_{\ell =\T, \neg\ell=\T}}$
$\wedge$ $\neg \varphi_{D|_{\ell=\T, \neg\ell=\F}}$ $\wedge$ $\neg \varphi_{D|_{\ell=\F, \neg\ell=\T}}$ $\wedge$ $\neg \varphi_{D|_{\ell=\F, \neg\ell=\F}}$.  To ensure that $CurrS$ in the current iteration does not include any such set obtained in previous iterations of the loop, we conjoin $D$ with all clauses in $CurrS$ after
dropping $\ell$.  The sets $CurrS$ obtained in each iteration of the repeat-until
loop are collected in $AllS$. 
Finally when $\ell$ becomes $\land$-unrealizable in circuit $D$, we obtain a satisfiable minimal hitting set (set cover) $HitS$ of $AllS$, i.e. a  subset of clauses that is jointly satisfiable with $\ell$ set to $\F$ and that includes a clause from every set in $AllS$.  Once $HitS$ is obtained, we exclude all $\ell$-leaves that appear in the clauses in $HitS$ to obtain a (maximal) $\land-$unrealizable subset of $\ell$-leaves, as required. Let us now formalize the correctness and complexity for this algorithm.

\begin{lemma}\label{lem:alg-getsaunf}
Algorithm {\tt GetSAUNF} returns a \eqnf\ circuit $F$ semantically equivalent to the input circuit $G$, with a worst-case running time exponential in $|G|$, and the worst-case size of $F$ also
exponential in $|G|$.
\end{lemma}
\begin{proof}
Since the circuit $G'$ represents a conjunction of a subset of clauses, each of which contains the literal $\ell$, the circuit $D = G|_{U:\T}$ obtained in line 6 of Algorithm {\tt GetSAUNF} simply represents the conjunction of all remaining clauses.  In the circuit $F$ constructed in line $11$ of Algorithm \ref{alg:getsaunf}, the set $U_\ell$ containing only the leaf $\ell$ is $\land-$unrealizable as it meets up $\neg{\ell}$ at an $\lor$ gate. In $F_{|U_\ell:\T}$, we can show that the set $U$ of $\ell-$leaves of $G'$ is $\land-$unrealizable as it was already $\land-$unrealizable in the circuit $G$. \\
\indent For $U$ to be $\land$-realizable in $F|_{U_\ell:\T}$, there must be an assignment $\sigma: \mset{\XX}\setminus\{v_\ell\} \cup \mset{\II} \rightarrow \{\F, \T\}$ such that $\llbracket \varphi_{G'}\rrbracket$ evaluates to $\ell$, $\llbracket \varphi_{G_1}\rrbracket$ to $\F$ and $\llbracket \varphi_{G_2} \rrbracket$ to $\T$.  However, if this were possible, then $U$ would have $\land$-realizable in $G$ (using the same assignment $\sigma$). However, this is a contradiction, since $U$ is $\land$-unrealizable in $G$.
\indent In $F|_{U_\ell:\T, U:\T}$, all $\ell$-leaves of $F$ have been re-labeled to $\T$, and therefore we can set the only $\neg{\ell}$-leaf (in set $U_{\neg\ell}$) to $\T$.  Finally, $F|_{U_\ell:\T, U:\T, U_{\neg\ell}:\T}$ = $G_1\lor G_2$, which is in \eqnf\ assuming the recursive calls return correct representations and using Lemma \ref{disjunction}. The base condition is when $G$ is independent of $X$, for which  it already returns the correct value. 

Note that the sequence of subsets returned in line 15 of Algorithm {\tt GetSAUNF} is $(U_\ell, U, U_{\neg\ell}, U_1, U_2)$, which corresponds to the sequence of setting leaves to $\T$ as discussed above.  The algorithm would be correct if it had returned $(U_\ell, U, U_{\neg\ell}, U_2, U_1)$ as the sequence of subsets of leaves as well.

In the worst-case, Algorithm {\tt GetSAUNF} can reduce to brute-force Shannon expansion if $U$ computed in line 4 of the algorithm always returns the empty set of leaves.  In this case,
both the running time and the size of $F$ can grow exponentially with $|G|$.
\end{proof}

\begin{algorithm}[t]
  \caption{\getSubset($G, \ell$)}
    \label{alg:getSubset}
        \KwIn{$G$: circuit representing {\CNF} formula, $\ell$: Literal}
    \KwOut{$T$: $\land-$unrealizable subset of $\ell$-leaves in $G$}
    $AllS$ := $CurrS$ := $\emptyset$\;
    $D$ := $G$\;
    \Repeat{$\ell$ is $\land$-unrealizable in $D$}
    {
      $\sigma$ := GetAssignment($D, \ell$)
      \tcp*{Assignment of vars except
      $v_\ell$ for which $\ell$ is $\land-$realizable in $D$}
      $CurrS$ := GetClausesEvaluatingToL($G, \sigma, \ell$)
      \tcp*{Set of all clauses (of  {\CNF} formula represented by $G$) containing $\ell$ that do not become $\T$ under $\sigma$}
      $AllS$ := $AllS \cup \{CurrS\}$\;
      $D$ := $D \wedge$ DisjoinWithoutLit($CurrS, \ell$)
      \tcp*{DisjoinWithoutLit($CurrS,\ell$) gives
      disjunction of clauses in $CurrS$ after dropping $\ell$}
  }
  $HitS$ := SatisfiableHittingSet($AllS$)\; 
  $T$ := Set of all $\ell$-leaves of $G$ that don't feed into any clause in $HitS$\;
  \Return $T$\;
\end{algorithm}

It remains to discuss the sub-routine {\tt GetSubset}.   The pseudo-code for this sub-routine is shown in Algorithm \ref{alg:getSubset}.  Let $v_\ell$ denote the underlying variable of the literal $\ell$.   $CurrS$ is a subset of clauses containing $\ell$ such that there exists an assignment $\sigma:\mset{\XX}\setminus\{v_\ell\} \cup \mset{\II} \rightarrow \{\F,\T\}$ for which all (and only) these clauses of the underlying {\CNF} formula evaluate to $\ell$, and $\ell$ is $\land-$realizable in $G$ under $\sigma$ (see Definition~\ref{def:l-unreal}).  Such an assignment $\sigma$ can be obtained by effectively finding a satisfying assignment of $\forall w \forall w'\, (\varphi_{D|_{\ell=w, \neg\ell= w'}} \Leftrightarrow (w \wedge w'))$, and therefore $\varphi_{D|_{\ell =\T, \neg\ell=\T}}$
$\wedge$ $\neg \varphi_{D|_{\ell=\T, \neg\ell=\F}}$ $\wedge$ $\neg \varphi_{D|_{\ell=\F, \neg\ell=\T}}$ $\wedge$ $\neg \varphi_{D|_{\ell=\F, \neg\ell=\F}}$.  To ensure that $CurrS$ in the current iteration does not include any such set obtained in previous iterations of the loop, we conjoin $D$ with the clause returned by DisjoinWithoutLit($CurrS, \ell$).  All sets $CurrS$ obtained as the repeat-until
loop iterates are collected in $AllS$. 
Finally when $\ell$ becomes $\land$-unrealizable in circuit $D$, we obtain a satisfiable minimal hitting set (set cover) $HitS$ of $AllS$, i.e. a  subset of clauses that is jointly satisfiable with $\ell$ set to $\F$ and that includes a clause from every set in $AllS$.  Given $AllS$, finding a minimal $HitS$ can be reduced to a MaxSAT problem.  Once $HitS$ is as obtained, we exclude all $\ell$-leaves that appear in the clauses in $HitS$ to obtain a (maximal) $\land-$unrealizable subset of $\ell$-leaves. With this, we can  state the correctness and complexity of Algorithm~\ref{alg:getSubset}.
\begin{lemma}\label{lem:alg2}
Algorithm {\tt GetSubset} returns a $\land-$unrealizable subset of $\ell$-leaves of $G$,  and takes  worst-case time exponential in $|G|$.
\end{lemma}
\begin{proof}
Suppose Algorithm {\tt GetSubset} returned an $\land$-realizable subset of $\ell$-leaves of $G$.  Let $\sigma$
be the corresponding assignment of $\mset{\XX}\setminus \{v_\ell\} \cup \mset{\II}$. The set $S$ of clauses of $\varphi_G$ that do not become $\T$ under $\sigma$ must then be either equal to or a superset of $CurrS$ in some iteration of the repeat-until loop of lines $3$-$7$.  Therefore, $HitS$ must include some clause from $S$.  This implies that $T$ contains at least one $\ell$-labeled leaf that feeds into a clause in $HitS$ -- a contradiction.

The worst-case running time is dominated by the product of the number of times the repeat-until loop of lines 3--8 iterates and the time required to obtain $\sigma$ (line 4) and check the loop termination condition (line 8).  The count of loop iterations can be as high as the count of all subsets of $\ell$-labeled leaves. This is exponential in $|G|$ in the worst-case. Computing $\sigma$ and checking the loop termination condition also require time exponential in $|G|$ in the worst-case.
Hence, the worst-case running time of Algorithm {\tt GetSubset} is exponential in $|G|$.
\end{proof}
Note that we can modify the loop termination condition in Algorithm {\tt GetSubset} by incorporating a timeout. In case a timeout happens, we conservatively return $\emptyset$ as $T$. This reduces the worst-case running time, providing a tradeoff between running time and precision of computation.


\section{Some interesting applications}\label{sec:appl}
In Section~\ref{sec:intro}, we  described the $n$-bit \emph{factorization problem} -- a problem of immense interest in cryptanalysis.
We now show some interesting partial results using \eqnf\ circuits. We start with a relational specification $R$ over inputs $\mseq{I}$ and outputs $(\mseq{X},\mseq{Y})$ defined by $(\mseq{X}\times_{[n]} \mseq{Y} = \mseq{I})$ where $\times_{[n]}$ denotes $n$-bit unsigned integer multiplication. This specification evaluates to $\ttrue$ iff the given product relation holds, where $\mseq{X}, \mseq{Y}$ are $n$-bit output vectors and $\mseq{I}$ is a $2n$ bit input vector.  For $1\leq l \leq j \leq 2n$, we define a parametrized specification $\R{l}{j}$ over $\mseq{X}$, $\mseq{Y}$ and $\mseq{I}$  that evaluates to $\ttrue$ if and only if the bits from position $l$ to $j$ of $\mseq{X}\times_{[n]} \mseq{Y}$ match the corresponding bits of $\mseq{I}$.

Let $\mathbf{1}$ denote an $n$-bit representation of the integer $1$. If $\R{1}{2n}\wedge(\mseq{X} \neq {\bf 1}) \wedge (\mseq{Y}\neq {\bf 1})$ can be represented as a polynomial (in $n$) sized {\eqnf} circuit, our results show that Skolem functions of size polynomial in $n$ can be obtained for $n$-bit factorization with non-trivial (i.e. $\neq 1$) factors. This would have serious ramifications for cryptanalysis.  While we are not close to achieving such a result, our initial studies show some interesting results in trying to represent $\R{l}{j}$ in \eqnf.  Note that it is already known from~\cite{bryant91} that representing $\R{n}{n}$ requires exponentially large {\ROBDD}s, and sub-exponential representations using {\DNNF}, {\dDNNF}, {\wDNNF} or {\snf} are not known.   With \eqnf\ circuits however, we obtain a significant improvement.
\begin{theorem}\label{theorem:gen}
$For \ l\leq j\leq 2n, j-l<n, \R{l}{j}$ is representable by a polynomial (in $n$) sized {\eqnf} circuit.
\end{theorem}
\begin{proof}
By Theorem~\ref{thm:characterize}, a polynomial-sized {\eqnf} circuit can be generated from a polynomial-sized Skolem function vector.
Therefore, we focus on obtaining a polynomial-sized Skolem function vector for $\R{l}{j}, \ l\leq j\leq 2n, j-l<n$.
We use the notation $\mseq{X}_l$ to denote the $l^{th}$ least significant bit of $\mseq{X}$ and $\mseq{X}_{l,j}$ to denote the bit-slice of $\mseq{X}$ from $l$ to $j$ (both included).
Using similar notations for $\mseq{Y}$ and $\mseq{I}$, we consider two cases:
\begin{itemize}
\item $l\leq n$: A Skolem function vector for $(\mseq{X}, \mseq{Y})$ is given by
$\psi^{x_k} = \F$ (resp \T), if $k \neq l$ (resp. $k=l$), and $\psi^{y_k} = \II_{k+l-1}$ (resp. $\F$) if $k \le j+1-l$ (resp. $k > j+1-l$). 

  Intuitively, $\mseq{X}$ represents  $2^{l-1}$ and $\mseq{Y}=\mseq{I}_{l,j}$. 
\item $l>n$:  Following similar logic, the Skolem function vector is given by $\psi^{x_k} = \F$ (resp. $\T$) if $k \neq n$ (resp. $k = n$), and
$\psi^{y_k} = \II_{k+n-1}$ (resp. $\F$) if $l-n ~<~ k ~\le~ j+1-n$ (resp. otherwise).

  Intuitively, $\mseq{X}$ represents $2^{n-1}$ and $\mseq{Y}$ represents $2^{l-n}\times_{[n]} \mseq{I}_{l,j}$.
\end{itemize}
Having generated a polynomial-sized Skolem function vector for $\R{l}{j}, j-l<n$, we can generate a corresponding polynomial-sized {\eqnf} circuit using Theorem~\ref{thm:characterize}.
\end{proof}

Surprisingly, we can use Theorem \ref{theorem:gen} to also show that a restricted version of division has a polynomial sized \eqnf\  representation.  Consider the same relational specification $\mseq{X} \times_{[n]} \mseq{Y} =\mseq{I}$ considered earlier.  For the division problem, we treat $(\mseq{I},\mseq{Y})$ as system inputs and $\mseq{X}$ as system outputs, and write the relation as $\mseq{X} = \mseq{I}/\mseq{Y}$ and obtain the following theorem. 

\begin{theorem}
The relation $\mseq{X} = \mseq{I}/\mseq{Y}$, with inputs $\mseq{I},\mseq{Y}$ restricted to odd numbers (i.e. the relation evaluates to $\F$ if $\mseq{I}$ or $\mseq{Y}$ is even), is representable as a polynomial (in $n$) sized  \eqnf\ circuit.
\end{theorem}
\begin{proof}
For notational convenience,  we use $\XX$, $\II$ and $\YY$ to denote both sequences of Boolean variables, and also the unsigned integers represented by the corresponding bit-vectors.  We use $\times$ instead of $\times_{[n]}$ to denote $n$-bit unsigned integer multiplication, and $+$ (resp. $-$) to denote $n$-bit unsigned integer addition (resp. subtraction).

We give below a polynomial-sized Skolem function vector for division, which can be used to obtain a \eqnf\ form by Theorem~\ref{thm:characterize}.  

 Suppose we have inputs $\II$ and $\YY$ and we have to find $\XX$ such that $\XX\times \YY = \II$. We first show that for odd valued inputs, if $(\XX\times \YY)mod\ 2^n= \II\ mod\ 2^n$ and if $\II$ is divisible by $\YY$ then $\XX\times \YY = \II$. Suppose there are two values $\XX^1, \XX^2$ such that $(\XX^1\times \YY)\ mod\ 2^n= \II\ mod\ 2^n$ and $(\XX^2\times \YY)\ mod\ 2^n= \II\ mod\ 2^n$. Then  $(\XX^1-\XX^2)\times \YY\equiv ~0\ mod\ 2^n$. However,  $\YY$ is odd and therefore co-prime to $2^n$; hence $(\XX^1-\XX^2)\equiv ~0\ mod\ 2^n$. Since, $\XX^1, \XX^2 < 2^n$, we must have $\XX^1= \XX^2$. Therefore, the generated $\XX$ from the Skolem function vector is correct if there exists a solution that matches the least significant $n$ bits of $\II$. Using the notation defined above, denote $\XX_l$ to be the $l^{th}$ least significant bit of $\XX$ and $\XX_{l,j}$ to denote the bit vector from $\XX_l$ to $\XX_j$ (both included). Now, since the inputs $\II$ and $\YY$ are odd, $\YY_1 = \T, \II_1 = \T$. Therefore $\psi^{\XX_1} = \T$. \\
 \indent Now note that $(\XX\times \YY)_i = (\XX_{1,i}\times \YY_{1,i})_i = (\XX_{1,i-1}\times \YY_{1,i} + \XX_i\cdot\YY_1\times \big(2^{i-1}\big)_i = \XX_i\cdot\YY_1\xor(\XX_{1,i-1}\times \YY_{1,i})_i$ (using the structure of multiplication), where "$\cdot$" denotes $1$-bit multiplication and $\xor$ denotes $1$-bit addition modulo $2$. Therefore, $\II_i = \XX_i\xor (\XX_{1,i-1}\times \YY_{1,i})_i$, or equivalently, $\XX_i = \II_i \xor (\XX_{1,i-1}\times \YY_{1,i})_i$. 
 
 It is now easy to see that once we obtain a Skolem function for $\psi^{\XX_1}$ to $\psi^{\XX_{i-1}}$ in this manner, we can recursively generate the skolem function for $\XX_i$, giving the entire Skolem function vector for $\XX$.
\end{proof}
Note that while we have used a specific Skolem function vector above, once  the \eqnf\ form is obtained, it can be used to generate other Skolem function vectors as well (from Algorithm ~\ref{alg:skgen}).

\section{Conclusion}\label{sec:concl}
In this paper, we presented a normal form for Boolean relational specifications that characterizes efficient Skolem function synthesis. This is a significantly stronger characterization than those used in earlier works.  \eqnfs\ is exponentially more succinct than {\DNNF}, {\dDNNF} while enjoying similar composability properties. It also strictly subsumes the recently proposed \snf.  As future work, we plan to improve the compilation algorithm and apply it to challenging benchmarks. It would also be interesting to see if similar characterizations or normal forms exist and are efficiently computable for Skolem functions for first order logic~\cite{AC21-arxiv}, i.e., beyond the propositional case that we treated in this work.

\bibliography{ref}
\bibliographystyle{IEEEtran}
\end{document}